\documentclass[journal,onecolumn,11pt,twoside]{IEEEtran}
\usepackage{epsfig,graphics,subfigure}
\usepackage{amsmath,amssymb,amsthm,multirow,balance,dsfont,hyperref,cases}
\usepackage{cite,cleveref}
\usepackage[table]{xcolor}
\usepackage{multirow,array}
\usepackage[ruled,vlined]{algorithm2e}
\usepackage{colortbl}
\usepackage{tabularx}

\usepackage{scalerel,stackengine}
\stackMath
\newcommand\widecheck[1]{%
\savestack{\tmpbox}{\stretchto{%
  \scaleto{%
    \scalerel*[\widthof{\ensuremath{#1}}]{\kern-.3pt\bigwedge\kern-.3pt}%
    {\rule[-\textheight/3]{1ex}{\textheight}}
  }{\textheight}%
}{0.5ex}}%
\stackon[1pt]{#1}{\scalebox{-0.8}{\tmpbox}}%
}

\newcommand{\vertiii}[1]{{\left\vert\kern-0.25ex\left\vert\kern-0.25ex\left\vert #1 
    \right\vert\kern-0.25ex\right\vert\kern-0.25ex\right\vert}}
\newtheorem{theorem}{Theorem}
\newtheorem{lemma}{Lemma}

\newtheorem{assumption}{Assumption}

\usepackage{color}
\def\cblue{\textcolor{black}}
\def\cred{\textcolor{black}}

\newcommand{\ba}{\boldsymbol{a}}
\newcommand{\bc}{\boldsymbol{c}}
\newcommand{\bd}{\boldsymbol{d}}
\newcommand{\bs}{\boldsymbol{s}}
\newcommand{\bu}{\boldsymbol{u}}
\newcommand{\bv}{\boldsymbol{v}}
\newcommand{\bw}{\boldsymbol{w}}
\newcommand{\bz}{\boldsymbol{z}}
\newcommand{\bx}{\boldsymbol{x}}
\newcommand{\by}{\boldsymbol{y}}

\newcommand{\bpsi}{\boldsymbol{\psi}}

\newcommand{\bH}{\boldsymbol{H}}

\newcommand{\cA}{\mathcal{A}}
\newcommand{\cB}{\mathcal{B}}
\newcommand{\cC}{\mathcal{C}}

\newcommand{\cD}{\mathcal{D}}

\newcommand{\cF}{\mathcal{F}}
\newcommand{\cFb}{\overline{\mathcal{F}}}
\newcommand{\cS}{\mathcal{S}}
\newcommand{\cR}{\mathcal{R}}

\newcommand{\cH}{\mathcal{H}}
\newcommand{\cI}{\mathcal{I}}
\newcommand{\cJ}{\mathcal{J}}
\newcommand{\cK}{\mathcal{K}}
\newcommand{\cL}{\mathcal{L}}

\newcommand{\cN}{\mathcal{N}}
\newcommand{\cP}{\mathcal{P}}

\newcommand{\cT}{\mathcal{T}}
\newcommand{\cU}{\mathcal{U}}
\newcommand{\cu}{{\scriptscriptstyle\mathcal{U}}}
\newcommand{\cY}{\mathcal{Y}}
\newcommand{\cV}{\mathcal{V}}
\newcommand{\cX}{\mathcal{X}}

\newcommand{\cw}{{\scriptstyle\mathcal{W}}}

\newcommand{\ccw}{{\scriptscriptstyle\mathcal{W}}}

\newcommand{\bcB}{\boldsymbol{\cal{B}}}
\newcommand{\bcH}{\boldsymbol{\cal{H}}}
\newcommand{\bcHt}{\widetilde{\boldsymbol{\cal{H}}}}
\newcommand{\bcD}{\boldsymbol{\cal{D}}}
\newcommand{\bcF}{\boldsymbol{\cal{F}}}
\newcommand{\bcw}{\boldsymbol{\cw}}

\newcommand{\bwt}{\widetilde \bw}

\newcommand{\bcwt}{\widetilde\bcw}
\newcommand{\bcwb}{\overline\bcw}
\newcommand{\bcwc}{\widecheck{\bcw}}

\newcommand{\wb}{\overline{w}}

\newcommand{\cBb}{\overline{\mathcal{B}}}
\newcommand{\bsb}{\overline{\bs}}
\newcommand{\bsc}{\widecheck{\bs}}

\newcommand{\expec}{\mathbb{E}}

\newcommand{\col}{\text{col}}
\newcommand{\card}{\text{card}}
\newcommand{\prox}{\text{prox}}

\newcommand{\bvc}{\text{bvec}}
\newcommand{\vc}{\text{vec}}
\newcommand{\tr}{\text{Tr}}
\newcommand{\diag}{\text{diag}}

\DeclareMathOperator*{\minimize}{minimize}
\DeclareMathOperator*{\st}{subject~to}

\newcolumntype{C}[1]{>{\centering\arraybackslash}m{#1}}
\begin{document}
\title{Adaptation and learning over networks under  subspace constraints -- Part II: Performance Analysis}
\author{Roula Nassif$^\dag$, \IEEEmembership{Member, IEEE},  Stefan Vlaski$^{\dag,\ddag}$, \IEEEmembership{Student Member, IEEE}, \\
Ali H. Sayed$^\dag$, \IEEEmembership{Fellow Member, IEEE}\\
\thanks{This work was supported in part by NSF grant CCF-1524250. A short conference version of this work appears in~\cite{nassif2019distributed}.}

\vspace{0.5cm}
\small{\linespread{0.2}$^\dag$ Institute of Electrical Engineering, EPFL, Switzerland}\\
\vspace{0.1cm}
\small{\linespread{0.2}$^\ddag$ Electrical Engineering Department, UCLA, USA}\\
\vspace{0.3cm}
roula.nassif@epfl.ch\hspace{0.5cm} stefan.vlaski@epfl.ch\hspace{0.5cm} ali.sayed@epfl.ch}

%

\maketitle

\begin{abstract}
Part~I of this paper considered optimization problems over networks where agents have individual objectives to meet, or individual parameter vectors to estimate, subject to subspace constraints that \cblue{require} the objectives across the network  to lie in  low-dimensional \cblue{subspaces}. \cblue{Starting from} the centralized projected gradient descent, an iterative and distributed solution was proposed that responds to streaming data and employs stochastic approximations in place of actual gradient vectors, which are generally unavailable. We examined the second-order stability of the learning algorithm and we showed that, for small step-sizes $\mu$, the proposed strategy leads to small estimation errors on the order of $\mu$. 
This Part~II examines steady-state performance. The results reveal explicitly the influence of the gradient noise, data characteristics, and subspace constraints, on the network performance. The results also show that in the small step-size regime, the iterates generated by the distributed algorithm achieve the centralized steady-state performance. \end{abstract}


\newpage

\section{Introduction}
As pointed out in Part~I~\cite{nassif2019adaptation} of this work, most prior literature on distributed inference over networks focuses on \textit{consensus} problems, where agents with separate objective functions need to agree on a common parameter vector corresponding to the minimizer of the aggregate sum of individual costs~\cite{sayed2013diffusion,sayed2014adaptation,chen2013distributed,zhao2015asynchronous,vlaski2019diffusion,rabbat2005quantized,kar2009distributed,srivastava2011distributed,olfati2007consensus,ram2010distributed}.  In this paper, and its accompanying Part~I~\cite{nassif2019adaptation}, we focus instead on multitask networks where the agents need to estimate and track multiple objectives simultaneously~\cite{plata2017heterogeneous,chen2015diffusion,nassif2016proximal,eksin2012distributed,hassani2017multi,mota2015distributed,zhao2014distributed,barbarossa2009distributed}. Based on the type of prior information that may be available about how the tasks are related to each other, multitask algorithms can be derived by translating the prior information into constraints on the parameter vectors to be inferred.

In this paper, and the accompanying Part~I~\cite{nassif2019adaptation}, we consider multitask inference problems where each agent seeks to minimize an individual cost, and where the collection of parameter vectors to be estimated across the network is required to lie in a low-dimensional subspace.  That is, we let $w_k\in{\mathbb{C}^{M_k}}$ denote some parameter vector at node $k$ and let $\cw=\col\{w_1,\ldots,w_N\}$ denote the collection of parameter vectors from across the network ($N$ is the number of agents in the network). We associate with each agent $k$ a differentiable convex cost $J_k(w_k):{\mathbb{C}^{M_k}}\rightarrow\mathbb{R}$, which is expressed as the expectation of some loss function $Q_k(\cdot)$ and written as $J_k(w_k)=\expec Q_k(w_k;\bx_k)$, where $\bx_k$ denotes the random data. The expectation is computed over the distribution of the data. Let $M=\sum_{k=1}^NM_k$. We consider constrained problems of the  form: 
\begin{equation}
\label{eq: constrained optimization problem}
\begin{split}
\cw^o=&\,\arg\min_{\ccw}~~ J^\text{glob}(\cw)\triangleq\sum_{k=1}^N J_k(w_k),\\
&\,\st~~\cw\in\cR(\cU),
\end{split}
\end{equation}
where $\cR(\cdot)$ denotes the range space operator, and $\cU$ is an $M\times P$ full-column rank matrix with  $P\ll M$. 
Each agent $k$ is interested in estimating  the $k$-th $M_k\times 1$ subvector $w^o_k$ of $\cw^o=\col\{w^o_1,\ldots,w^o_N\}$. 

In order to solve problem~\eqref{eq: constrained optimization problem}, we proposed in Part~I~\cite{nassif2019adaptation} the following adaptive and distributed strategy:
\begin{equation}
\label{eq: distributed solution}
\left\lbrace
\begin{array}{rl}
\bpsi_{k,i}=&\hspace{-2mm}\bw_{k,i-1}-\mu\widehat{\nabla_{w_k^*}J_k}(\bw_{k,i-1}),\\
\bw_{k,i}=&\hspace{-2mm}\sum\limits_{\ell\in\cN_k}A_{k\ell}\bpsi_{\ell,i},
\end{array}
\right.
\end{equation}
where $\mu>0$ is a small step-size parameter, $\bpsi_{k,i}$ is an intermediate estimate, $\bw_{k,i}$ is the estimate of $w^o_k$ at agent $k$ and iteration $i$, $\cN_k$ denotes the neighborhood of agent $k$, and ${\nabla_{w_k^*}J_k}(\cdot)$ is the (Wirtinger) complex gradient~\cite[Appendix~A]{sayed2014adaptation} of $J_k(\cdot)$ relative to $w^*_k$ (complex conjugate of $w_k$). Notice that approximate gradient vectors $\widehat{\nabla_{w_k^*}J_k}(\cdot)$ are employed in~\eqref{eq: distributed solution} instead of true gradient vectors ${\nabla_{w_k^*}J_k}(\cdot)$ since we are interested in solving~\eqref{eq: constrained optimization problem} in the stochastic setting when the distribution of the data $\bx_k$ is unknown. A common construction in stochastic approximation theory is to employ the following approximation at iteration $i$:
\begin{equation}
\label{eq: stochastic approximation of the gradient}
\widehat{\nabla_{w_k^*}J_k}(w_k)={\nabla_{w_k^*}Q_k}(w_k;\bx_{k,i}),
\end{equation}
where $\bx_{k,i}$ represents the data observed at iteration $i$. The difference between the true gradient and its approximation is called  the {gradient noise} $\bs_{k,i}(\cdot)$:
\begin{equation}
\label{eq: gradient noise process}
\bs_{k,i}(w)\triangleq\nabla_{w^*_k}J_k(w)-\widehat{\nabla_{w^*_k}J_k}(w).
\end{equation}
 This noise will seep into the operation of the algorithm and one main challenge is to show that despite its presence, agent $k$ is still able to approach $w^o_k$ asymptotically.  The matrix $A_{k\ell}$ appearing in~\eqref{eq: distributed solution} is of size $M_k\times M_\ell$. It multiplies the intermediate estimate $\bpsi_{\ell,i}$ arriving from neighboring agent $\ell$ to agent $k$. Let $\cA\triangleq [A_{k\ell}]\in\mathbb{C}^{M\times M}$ denote the matrix that collects all these blocks. This $N\times N$ block matrix is chosen by the designer to satisfy the following two conditions:
\begin{numcases}{}
  \lim\limits_{i\rightarrow\infty}\cA^i=\cP_{\cu}, & \label{eq: condition 1 on A}\\
  A_{k\ell}=[\cA]_{k\ell}=0, \quad\text{if }\ell\notin\cN_k \text{ and } k\neq\ell,& \label{eq: condition 2 on A}
\end{numcases}
where $[\cA]_{k\ell}$ denotes the $(k,\ell)$-th block of $\cA$ and $\cP_{\cu}$ is the projector onto the $P$-dimensional subspace of $\mathbb{C}^M$ spanned by the columns of $\cU$:
\begin{equation}
\label{eq: P_R(U)}
\cP_{\cu}=\cU(\cU^*\cU)^{-1}\cU^*.
\end{equation} 
 The sparsity condition~\eqref{eq: condition 2 on A} characterizes the network topology and ensures local exchange of information at each time instant $i$. It is shown in Part~I~\cite{nassif2019adaptation} that the matrix equation~\eqref{eq: condition 1 on A} holds, if and only if, the following conditions on the projector $\cP_{\cu}$ and the matrix $\cA$ are satisfied:
\begin{align}
\cA \,\cU&=\cU,\label{eq: first condition on eigenvector}\\
\cU^*\cA&=\cU^*.\label{eq: second condition on eigenvector}\\
\rho(\cA-\cP_{\cu})&<1,\label{eq: third condition in the proposition}
\end{align}
where $\rho(\cdot)$ denotes the spectral radius of its matrix argument. Conditions~\eqref{eq: first condition on eigenvector} and~\eqref{eq: second condition on eigenvector} state that the $P$ columns of $\cU$ are right and left eigenvectors of $\cA$ associated with the eigenvalue $1$. Together with these two conditions, condition~\eqref{eq: third condition in the proposition} means that $\cA$ has $P$ eigenvalues at one, and that all other eigenvalues are strictly less than one in magnitude. Combining conditions~\eqref{eq: first condition on eigenvector}--\eqref{eq: third condition in the proposition} with the sparsity condition~\eqref{eq: condition 2 on A}, the design of a matrix $\cA$ to run~\eqref{eq: distributed solution} can be written as the following feasibility problem:
\begin{equation}
\label{eq: feasibility problem A}
\begin{array}{cl}
\text{find}&\cA\\
\text{such that}&\cA\,\cU=\cU,~~\cU^*\cA=\cU^*,\\ 
&\rho(\cA-\cP_{\cu})< 1,\\ 
&[\cA]_{k\ell}=0, ~\text{if }\ell\notin\cN_k \text{ and }\ell\neq k.
\end{array}
\end{equation} 
Not all network topologies satisfying~\eqref{eq: condition 2 on A} guarantee the existence of an $\cA$ satisfying condition~\eqref{eq: condition 1 on A}. The higher the dimension of the signal subspace is, the greater the graph connectivity has to be. In the works~\cite{nassif2019distributed,barbarossa2009distributed}, it is assumed that the sparsity constraints~\eqref{eq: condition 2 on A} and the signal subspace lead to a feasible problem. That is, it is assumed that problem~\eqref{eq: feasibility problem A} admits at least one solution. As a remedy for the violation of such assumption, one may increase the network connectivity by increasing the transmit power of each node, i.e., adding more links~\cite{barbarossa2009distributed}. In Section~\ref{sec: finding a combination matrix A} of this part, we shall relax the feasibility assumption by considering the problem of finding an $\cA$ that minimizes the number of edges to be added to the original topology while satisfying the constraints~\eqref{eq: first condition on eigenvector},~\eqref{eq: second condition on eigenvector}, and~\eqref{eq: third condition in the proposition}. In this case, if the original topology leads to a feasible solution, then no links will be added. Otherwise, we assume that the designer is able to  add some links to make the problem feasible.

When studying the performance of algorithm~\eqref{eq: distributed solution} relative to $\cw^o$, we assume that a feasible $\cA$ (topology) is computed by the designer and that its blocks $\{A_{k\ell}\}_{\ell\in\cN_k}$ are provided to agent $k$ in order to run~\eqref{eq: distributed solution}. We carried out in Part~I~\cite{nassif2019adaptation} a detailed stability analysis of the proposed strategy~\eqref{eq: distributed solution}. We showed that, despite the gradient noise, the distributed strategy~\eqref{eq: distributed solution} is able to converge in the mean-square-error sense within $O(\mu)$ from the solution of the constrained problem~\eqref{eq: constrained optimization problem}, for sufficiently small step-sizes $\mu$. We particularly established that, for each agent $k$, the error variance relative to $w^o_k$ enters a bounded region whose size is in the order of $\mu$, namely, $\limsup_{i\rightarrow\infty}\expec\|w^o_k-\bw_{k,i}\|^2=O(\mu)$. In Section~\ref{sec: Stochastic performance analysis} of this Part~II, we will assess the size of this mean-square-error by deriving closed-form expression for the network mean-square-deviation (MSD) defined by~\cite{sayed2014adaptation}:
\begin{equation}
\label{eq: MSD definition}
\text{MSD}\triangleq \mu\lim_{\mu\rightarrow 0}\left(\limsup_{i\rightarrow\infty}\frac{1}{\mu}\expec\left(\frac{1}{N}\|\cw^o-\bcw_{i}\|^2\right)\right),
\end{equation}
where $\bcw_i\triangleq\col\{\bw_{k,i}\}_{k=1}^N$. In other words, we will assess the size of the constant multiplying $\mu$ in the $O(\mu)-$term. This closed form expression will reveal explicitly the influence of the data characteristics (captured by the second-order properties of the costs and second-order moments of the gradient noises) and subspace constraints (captured by $\cU$), on the network performance. In this way, we will be able to conclude that distributed strategies of the form~\eqref{eq: distributed solution} with small step-sizes are able to lead to reliable performance even in the presence of gradient noise. We will be able also to conclude that the iterates generated by the distributed implementation achieve the centralized steady-state performance. Particularly, we compare the performance of strategy~\eqref{eq: distributed solution} to the following centralized stochastic gradient projection algorithm~\cite{bertsekas1999nonlinear}:
\begin{equation}
\label{eq: centralized solution}
\bcw^{{c}}_i=\cP_{\cu}\left(\bcw^c_{i-1}-\mu\,\col\left\{\widehat{\nabla_{w_k^*}J_k}(\bw^c_{k,i-1})\right\}_{k=1}^N\right),\quad i\geq 0,
\end{equation}
where $\bcw^{{c}}_i=\col\{\bw^c_{1,i},\ldots,\bw^c_{N,i}\}$ is the estimate of $\cw^o$ at iteration $i$. Observe that each agent at each iteration needs to send its data to a fusion center, which performs the projection in~\eqref{eq: centralized solution}, and then sends the resulting estimates $\bw^c_{k,i}$ back to the agents. Finally, simulations will be provided in Section~\ref{subsec: theoretical model validation} to verify the theoretical findings.
\section{Stochastic performance analysis}
\label{sec: Stochastic performance analysis}
In Part~I~\cite{nassif2019adaptation}, we carried out a detailed stability analysis of the proposed strategy~\eqref{eq: distributed solution}. We showed, under some Assumptions on the risks $\{J_k(\cdot)\}$ and on the gradient noise processes $\{\bs_{k,i}(\cdot)\}$ defined by~\eqref{eq: gradient noise process}, that a network running strategy~\eqref{eq: distributed solution} with a matrix $\cA$ satisfying conditions~\eqref{eq: condition 2 on A},~\eqref{eq: first condition on eigenvector},~\eqref{eq: second condition on eigenvector}, and~\eqref{eq: third condition in the proposition} is mean-square-error stable for sufficiently small step-sizes, namely, it holds that:
\begin{equation}
\label{eq: mean-square error convergence result}
\limsup_{i\rightarrow\infty}\expec\|w^o_k-\bw_{k,i}\|^2=O(\mu),\quad k=1,\ldots,N,
\end{equation} 
for small enough $\mu$--see~\cite[Theorem~1]{nassif2019adaptation}. Expression~\eqref{eq: mean-square error convergence result} indicates that the mean-square error $\expec\|\cw^o-\bcw_i\|^2$ is on the order of $\mu$. In this section, we are interested in characterizing how close the $\bcw_i$ gets to the network limit point $\cw^o$. In particular, we will be able to characterize the network mean-square deviation (MSD) (defined by~\eqref{eq: MSD definition}) value in terms of the step-size $\mu$, the data-type variable $h$ defined in Table~\ref{table: variables table}, the  second-order properties of the costs (captured by $\cH^o$ defined below in~\eqref{eq: definition of H o general}), the second-order moments of the gradient noises (captured by $\cS$ defined below in~\eqref{eq: definition of S general}), and the subspace constraints (captured by $\cU^e$ defined in Table~\ref{table: variables table}) as follows: 
 \begin{equation}
\label{eq: final result for MSD}
\text{MSD}=\frac{\mu}{2hN}{\tr}\left(\left((\cU^e)^*\cH^o\cU^e\right)^{-1}\left((\cU^e)^*\cS\cU^e\right)\right).
\end{equation}
As explained in Part~I~\cite{nassif2019adaptation}, in the general complex data case, extended vectors and matrices need to be introduced in order to analyze the network evolution. The arguments and results presented in this section are applicable to both cases of real and complex data through the use of data-type variable $h$. Table~\ref{table: variables table} lists a couple of variables and symbols that will be used in the sequel for both real and complex data cases. The matrix $\cI$ in Table~\ref{table: variables table} is a permutation matrix of $2N\times 2N$ blocks with $(m,n)$-th block given by:
\begin{equation}
\label{eq: permutation matrix block}
[\cI]_{mn}\triangleq\left\lbrace
\begin{array}{ll}
I_{M_k},&\text{if } m=k, n=2(k-1)+1\\
I_{M_k},&\text{if } m=k+N, n=2k\\
0,&\text{otherwise}\\
\end{array}
\right.
\end{equation} 
for $m,n=1,\ldots,2N$ and $k=1,\ldots,N$.

\begin{table}
\caption{Definition of some variables used throughout the analysis. $\cI$ is a permutation matrix defined by~\eqref{eq: permutation matrix block}.}
\begin{center}
\begin{tabular}{c|c|c}
\hline
\hline
Variable&Real data case &Complex data case \\ \hline
Data-type variable $h$&1&2\\
Gradient vector&$\nabla_{w^\top_k}J_k(w_k)$&$\nabla_{w^*_k}J_k(w_k)$\\ 
Error vector $\bwt_{k,i}^e$&$\bwt_{k,i}$ from~\eqref{eq: error vector at node k}&$\left[\begin{array}{c}\bwt_{k,i}\\(\bwt_{k,i}^*)^\top \end{array}\right]$\\  [2.8ex]
Gradient noise $\bs_{k,i}^e(w)$&$\bs_{k,i}(w)$ from~\eqref{eq: gradient noise process}&$\left[\begin{array}{c}\bs_{k,i}(w)\\(\bs_{k,i}^*(w))^\top \end{array}\right]$\\ [2.8ex]
Bias vector $b_k^e$&$b_k$ from~\eqref{eq: bias vector}&$\left[\begin{array}{c}b_k\\(b_k^*)^\top \end{array}\right]$\\ [2.8ex]
$(k,\ell)$-th block of $\cA^e$ & $A_{k\ell}$&$\left[\begin{array}{cc}A_{k\ell}&0\\0&(A^*_{k\ell})^\top \end{array}\right]$\\[2.8ex]
Matrix $\cU^e$&$\cU$&$\cI^\top\left[\begin{array}{cc}\cU&0\\0&(\cU^*)^\top \end{array}\right]$\\ [2.8ex]
Matrix $\cJ^e_{\epsilon}$&$\cJ_{\epsilon}$ from~\eqref{eq: partitioning of cV-1}&$\left[
\begin{array}{cc}
\cJ_{\epsilon}&0\\
0&(\cJ^*_{\epsilon})^\top
\end{array}
\right]$\\ [2.8ex]
Matrix $\cV_{R,\epsilon}^e$&$\cV_{R,\epsilon}$ from~\eqref{eq: partitioning of cV-1}&\hspace{-1.5mm}$\cI^\top\left[
\begin{array}{cc}
\cV_{R,\epsilon}&0\\
0&(\cV_{R,\epsilon}^*)^\top
\end{array}
\right]$\\ [2.8ex]
Matrix $(\cV_{L,\epsilon}^e)^*$&$\cV_{L,\epsilon}^*$ from~\eqref{eq: partitioning of cV-1}&$\left[
\begin{array}{cc}
\cV_{L,\epsilon}^*&0\\
0&\cV_{L,\epsilon}^\top
\end{array}
\right]\cI$\\ [2.8ex]
Noise covariance $R_k^o$&$R_{q,k}$ from~\eqref{eq: constant covariance matrix},\eqref{eq: constant covariance matrix 2}&$\left[\begin{array}{cc}R_{s,k}&R_{q,k}\\R_{q,k}^*&R_{s,k}^\top \end{array}\right]$\\
\hline
 \end{tabular}
\end{center}
\label{table: variables table}
\end{table}


\subsection{Modeling assumptions from Part~I~\cite{nassif2019adaptation}}
In this section, we recall the assumptions used in Part~I~\cite{nassif2019adaptation} to establish the network mean-square error stability~\eqref{eq: mean-square error convergence result}. We first introduce the Hermitian Hessian \cblue{matrix} functions (see~\cite[Sec.~II-A]{nassif2019adaptation}):
\begin{align}
H_k(w_k)&\triangleq\nabla_{w_k}^2J_k(w_k),\qquad\qquad(hM_k\times hM_k)\\
\cH(\cw)&\triangleq\diag\left\{H_k(w_k)\right\}_{k=1}^N,\quad(hM\times hM).\label{eq: block diagonal Hessian}
\end{align}

\begin{assumption}{\emph{(Conditions on aggregate and individual costs).}}
\label{assump: risks}
The individual costs $J_k(w_k)\in\mathbb{R}$ are assumed to be twice differentiable and convex such that:
\begin{equation}
\label{eq: convexity condition}
\frac{\nu_k}{h}I_{hM_k}\leq H_k(w_k)\leq\frac{\delta_{k}}{h}I_{hM_k},
\end{equation} 
where $\nu_{k}\geq 0$ for $k=1,\ldots,N$. It is further assumed that, for any $\cw$, $\cH(\cw)$ satisfies:
\begin{equation}
\label{assump: aggregate risk}
0<\frac{\nu}{h}I_{hP}\leq(\cU^e)^*\cH(\cw)\cU^e\leq\frac{\delta}{h}I_{hP},
\end{equation}
for some positive parameters $\nu\leq \delta$. The data-type variable $h$ and the matrix $\cU^e$ are defined in \emph{Table~\ref{table: variables table}}.
\end{assumption}

\noindent As explained in~\cite{nassif2019adaptation}, condition~\eqref{assump: aggregate risk} ensures that problem~\eqref{eq: constrained optimization problem} has a unique minimizer $\cw^o$. 
\begin{assumption}{\emph{(Conditions on gradient noise).}}
\label{assump: gradient noise}
The gradient noise process defined in~\eqref{eq: gradient noise process} satisfies for any $\bw\in\bcF_{i-1}$ and for all $k,\ell=1,\ldots,N$:
\begin{align}
\expec[\bs_{k,i}(\bw)|\bcF_{i-1}]&=0,\label{eq: gradient noise mean condition}\\
\expec[\bs_{k,i}(\bw)\bs_{\ell,i}^*(\bw)|\bcF_{i-1}]&=0,\qquad k\neq \ell,\label{eq: uncorrelatedness of the gradient noises}\\
\expec[\bs_{k,i}(\bw)\bs_{\ell,i}^\top(\bw)|\bcF_{i-1}]&=0,\qquad k\neq \ell,\label{eq: circularity of the gradient noises}\\
\expec[\|\bs_{k,i}(\bw)\|^2|\bcF_{i-1}]&\leq(\beta_k/h)^2\|\bw\|^2+\sigma^2_{s,k},\label{eq: gradient noise mean square condition}
\end{align}
for some $\beta_k^2\geq 0$, $\sigma^2_{s,k}\geq 0$, and where $\bcF_{i-1}$ denotes the filtration generated by the random processes $\{\bw_{\ell,j}\}$ for all $\ell=1,\ldots,N$ and $j\leq i-1$. 
\end{assumption}

\begin{assumption}{\emph{(Condition on $\cU$).}}
\label{assump: matrix cU}
The full-column rank matrix $\cU$ is assumed to be semi-unitary, i.e., its column vectors are orthonormal and $\cU^*\cU=I_P$.
\end{assumption}

Consider the $N\times N$ block matrix $\cA^e$ whose $(k,\ell)$-th block is defined in Table~\ref{table: variables table}. This matrix will appear in our subsequent study. In~\cite[Lemma~2]{nassif2019adaptation}, we showed that this $hM\times hM$ matrix $\cA^e$ admits a  Jordan decomposition of the form:
\begin{equation}
\label{eq: jordan decomposition of A e}
\cA^e\triangleq\cV_{\epsilon}^e\Lambda_{\epsilon}^e(\cV_{\epsilon}^e)^{-1},
\end{equation}
with
\begin{equation}
\label{eq: partitioning of cVe-1}
\Lambda_{\epsilon}^e=\left[
\begin{array}{cc}
I_{hP}&0\\
0&\cJ^e_{\epsilon}
\end{array}
\right],~\cV^e_{\epsilon}=[\cU^e~\cV^e_{R,\epsilon}],~(\cV^e_{\epsilon})^{-1}=\left[
\begin{array}{c}
(\cU^e)^*\\
(\cV_{L,\epsilon}^e)^*
\end{array}
\right]
\end{equation}
where $\cU^e,\cJ^e_{\epsilon},\cV^e_{R,\epsilon}$, and $(\cV_{L,\epsilon}^e)^*$ are defined in {Table~\ref{table: variables table}} with the matrices $\cJ_\epsilon$, $\cV_{R,\epsilon}$, and $\cV_{L,\epsilon}^*$ originating from the eigen-structure of $\cA$. Under Assumption~\ref{assump: matrix cU}, the $M\times M$ combination matrix $\cA$ satisfying conditions~\eqref{eq: first condition on eigenvector},~\eqref{eq: second condition on eigenvector}, and~\eqref{eq: third condition in the proposition} admits a Jordan canonical decomposition of the form:
\begin{equation}
\label{eq: eigendecomposition of cA}
\cA\triangleq\cV_{\epsilon}\Lambda_{\epsilon}\cV_{\epsilon},
\end{equation}
with:
\begin{equation}
\Lambda_{\epsilon}=\left[
\begin{array}{cc}
I_{P}&0\\
0&\cJ_{\epsilon}
\end{array}
\right],~\cV_{\epsilon}=\left[
\begin{array}{cc}
\cU&\cV_{R,\epsilon}
\end{array}
\right],~
\cV_{\epsilon}^{-1}=\left[
\begin{array}{c}
\cU^*\\
\cV_{L,\epsilon}^*
\end{array}
\right],\label{eq: partitioning of cV-1}
\end{equation}
where $\cJ_{\epsilon}$ is a Jordan matrix with the eigenvalues (which may be complex but have magnitude less than one) on the diagonal and $\epsilon>0$ on the  super-diagonal. The eigen-decomposition~\eqref{eq: jordan decomposition of A e}  will be useful for establishing the mean-square performance.

The results in Part~I~\cite{nassif2019adaptation} established that the iterates $\bw_{k,i}$ converge in the mean-square error sense to a small $O(\mu)-$neighborhood around the solution $w^o$. In this part, we will be more precise and determine the size of this neighborhood, i.e., assess the size of the constant multiplying $\mu$ in the $O(\mu)-$term. To do so, we shall derive an accurate first-order expression for the mean-square error~\eqref{eq: mean-square error convergence result}; the expression will be accurate to first-order in $\mu$. 

To arrive at the desired expression, we start by motivating a long-term model for the evolution of the network error vector after sufficient iterations have passed, i.e., for $i\gg 1$. It turns out that the performance expressions obtained from analyzing the long-term model provide accurate expressions for the performance of the original network model to first order in $\mu$. To derive the long-term model, we follow the approach developed in~\cite{sayed2014adaptation}. The first step is to establish the asymptotic stability of the fourth-order moment of the error vector, $\expec\|w^o_k-\bw_{k,i}\|^4$.  Under the same settings of Theorem~1 in~\cite{nassif2019adaptation} with the second-order moment condition~\eqref{eq: gradient noise mean square condition} replaced by the fourth-order moment condition:
\begin{equation}
\label{eq: gradient noise fourth moment condition}
\expec\left[\|\bs_{k,i}(\bw)\|^4|\bcF_{i-1}\right]\leq(\beta_{4,k}/h)^4\|\bw\|^4+\sigma^4_{s4,k},
\end{equation}
with $\beta_{4,k}^4\geq 0$, $\sigma^4_{s4,k}\geq 0$, and using similar arguments as in~\cite[Theorem~9.2]{sayed2014adaptation}, we can show that the fourth-order moments of the network error vectors are stable for sufficiently small $\mu$, namely, it holds that {(see~Appendix~\ref{sec: stability of fourth-order error moment})} 
\begin{equation}
\label{eq: fourth order moment result}
\limsup_{i\rightarrow\infty}\expec\|w^o_k-\bw_{k,i}\|^4=O(\mu^2),\quad k=1,\ldots,N.
\end{equation} 
As explained in~\cite{sayed2014adaptation}, condition~\eqref{eq: gradient noise fourth moment condition} implies~\eqref{eq: gradient noise mean square condition}. We analyze the long-term model under the same settings of Theorem 1 in~\cite{nassif2019adaptation} and the following smoothness assumption on the individual costs.
\begin{assumption}{\emph{(Smoothness condition on individual costs).}}
\label{assumption: smoothness condition on Hessian}
It is assumed that each $J_k(w_k)$ satisfies the following smoothness condition close to the limit point $w^o_k$:
\begin{equation}
\label{eq: smoothness of the individual Hessian assumption}
\|\nabla_{w_k}J_k(w^o_k+\Delta w_k)-\nabla_{w_k}J_k(w^o_k)\|\leq\kappa_d\|\Delta w_k\|,
\end{equation}
for small perturbations $\|\Delta w_k\|$ and $\kappa_d\geq 0$.
\end{assumption}

\subsection{Long-term-error model}
\label{subsec: Long-term-error model}
To introduce the long-term model, we reconsider the network error recursion from Part~I~\cite{nassif2019adaptation}, namely, 
\begin{equation}
\label{eq: error recursion 1}
\fbox{$\bcwt_i^e=\bcB_{i-1}\bcwt^e_{i-1}-\mu\cA^e\bs_i^e+\mu\cA^e b^e$}
\end{equation}
where:
\begin{align}
\bcwt_i^e&\triangleq\col\left\{\bwt^e_{1,i},\ldots,\bwt^e_{N,i}\right\},\label{eq: cwt}\\
\bcH_{i-1}&\triangleq\diag\left\{\bH_{1,i-1},\ldots,\bH_{N,i-1}\right\},\\
\bcB_{i-1}&\triangleq\cA^e(I_{hM}-\mu\bcH_{i-1}),\label{eq: matrix Bi-1 1}\\
\bs_i^e&\triangleq\col\left\{\bs^e_{1,i}(\bw_{1,i-1}),\ldots,\bs^e_{N,i}(\bw_{N,i-1})\right\},\\
b^e&\triangleq\col\left\{b_1^e,\ldots,b_N^e\right\},\label{eq: expression of b 1}
\end{align}
where:
\begin{equation}
\bH_{k,i-1}\triangleq\int_{0}^1{\nabla_{w_k}^2J_k}(w^o_k-t\bwt_{k,i-1})dt,\label{eq: Hki-1}
\end{equation}
and $\bwt_{k,i}^e$, $\bs_{k,i}^e(\bw_{k,i-1})$, and $b_k^e$ are defined in Table~\ref{table: variables table} with:
\begin{align}
\bwt_{k,i}&\triangleq w^o_k-\bw_{k,i},\label{eq: error vector at node k}\\
b_k&\triangleq\nabla_{w^*_k}J_k(w^o_k).\label{eq: bias vector}
\end{align}
We rewrite~\eqref{eq: error recursion 1}  as:
\begin{equation}
\label{eq: long term error model 0}
\bcwt^e_i=\cB\bcwt^e_{i-1}-\mu\cA^e\bs^e_i+\mu\cA^e b^e+\mu\cA^e\bc_{i-1},
\end{equation}
in terms of the constant matrix $\cB$ and the random perturbation sequence $\bc_{i-1}$:
\begin{align}
\cB&\triangleq\cA^e(I_{hM}-\mu\cH^o),\label{eq: constant B}\\
\bc_{i-1}&\triangleq\bcHt_{i-1}\bcwt^e_{i-1},\label{eq: definition of the perturbation}
\end{align}
where $\cH^o$ and $\bcHt_{i-1}$ are given by:
\begin{align}
\bcHt_{i-1}&\triangleq\cH^o-\bcH_{i-1},\label{eq: definition of bcHt}\\
\cH^o&\triangleq\diag\{H_1^o,\ldots,H_N^o\},\label{eq: definition of H o general}
\end{align}
with each $H_k^o$ given by the value of the Hessian matrix at the limit point, namely,
\begin{equation}
H_k^o\triangleq \nabla_{w_k}^2J_k(w^o_k). \label{eq: H star k}
\end{equation}
By exploiting the smoothness condition~\eqref{eq: smoothness of the individual Hessian assumption}, and following an argument similar to~\cite[pp.~554]{sayed2014adaptation}, we can show from Theorem~1 in~\cite{nassif2019adaptation} that, for $i\gg 1$,  $\|\bc_{i-1}\|=O(\mu)$ with high probability. Motivated by this observation, we introduce the following approximate model, where the last term $\mu\cA^e\bc_{i-1}$ that appears in~\eqref{eq: long term error model 0}, which is $O(\mu^2)$, is removed:
\begin{equation}
\label{eq: long term error model 2}
\bcwt_i^{e'}=\cB\bcwt_{i-1}^{e'}-\mu\cA^e\bs^e_i(\bcw_{i-1})+\mu\cA^e b^e,
\end{equation}
for $i\gg 1$. Obviously, the iterates $\{\bcwt^{e'}_i\}$ generated by~\eqref{eq: long term error model 2} are generally different from the iterates  generated by  the original recursion~\eqref{eq: error recursion 1}. To highlight this fact, we are using the prime
notation for the state of the long-term model. Note that the driving process $\bs_{i}^e(\bcw_{i-1})$ in~\eqref{eq: long term error model 2} is the same gradient noise process from the original recursion~\eqref{eq: error recursion 1}.  

We start by showing that the mean-square difference between  $\{\bcwt_i^{e'},\bcwt_i^{e}\}$ is asymptotically bounded by $O(\mu^2)$ and that the mean-square-error of the long term model~\eqref{eq: long term error model 2} is within $O(\mu^{\frac{3}{2}})$ from the one of the original recursion~\eqref{eq: error recursion 1}.  Working with~\eqref{eq: long term error model 2} is much more tractable for performance analysis because its dynamics is driven by the constant matrix $\cB$ as opposed to the random matrix $\bcB_{i-1}$ in~\eqref{eq: error recursion 1}. Therefore, we shall work with model~\eqref{eq: long term error model 2} and evaluate its performance, which will provide an accurate representation for the performance of~\eqref{eq: distributed solution} to first order in $\mu$.

\begin{theorem}{\emph{(Size of approximation error).}} 
\label{theo: performance error}
Consider a network of $N$ agents running the distributed  strategy~\eqref{eq: distributed solution} with a matrix $\cA$ satisfying conditions~\eqref{eq: first condition on eigenvector},~\eqref{eq: second condition on eigenvector}, and~\eqref{eq: third condition in the proposition} and $\cU$ satisfying Assumption~\ref{assump: matrix cU}. Assume the individual costs, $J_k(w_k)$, satisfy the conditions in Assumptions~\ref{assump: risks} and~\ref{assumption: smoothness condition on Hessian}. Assume further that the gradient noise processes satisfy the conditions in Assumption~\ref{assump: gradient noise} with the second-order moment condition~\eqref{eq: gradient noise mean square condition} replaced by the fourth-order moment condition~\eqref{eq: gradient noise fourth moment condition}. Then, it holds that, for sufficiently small step-sizes:
\begin{align}
\limsup_{i\rightarrow\infty}\expec\|\bcwt_i^e-\bcwt^{e'}_i\|^2&=O(\mu^2),\label{eq: performance error result 1}\\
\limsup_{i\rightarrow\infty}\expec\|\bcwt_i^e\|^2&=\limsup_{i\rightarrow\infty}\expec\|\bcwt^{e'}_i\|^2+O(\mu^{3/2}).\label{eq: performance error result 2}
\end{align} 
\end{theorem}
\begin{proof}

See Appendix~\ref{app: Size of the approximation error}.
\end{proof}
Using similar eigenvalue perturbation arguments as in~\cite[Theorem~9.3]{sayed2014adaptation}, we can show that, under the same settings of Theorem~\ref{theo: performance error}, the constant matrix $\cB$ defined by~\eqref{eq: constant B} is stable for sufficiently small step-sizes  (see~Appendix~\ref{sec: Stability of of the coefficient matrix B}).
\subsection{Mean-square-error performance}
\label{subsec: Mean-square-error performance}
We showed in Theorem~1 in Part~I~\cite{nassif2019adaptation} that a network running the distributed strategy~\eqref{eq: distributed solution} is mean-square stable for sufficiently small $\mu$. Particularly, we showed that $\limsup_{i\rightarrow\infty}\expec\|w^o_k-\bw_{k,i}\|^2=O(\mu)$. In this section, we  assess the size of the mean-square error by measuring the network MSD defined by~\eqref{eq: MSD definition}.

We refer to the individual gradient noise \cblue{process} in~\eqref{eq: gradient noise process} and denote its conditional covariance matrix by:
\begin{equation}
\label{eq: conditional covariance matrix}
R_{s,k,i}^e(\bw)\triangleq\expec\left[\bs^e_{k,i}(\bw)\bs_{k,i}^{e*}(\bw)|\bcF_{i-1}\right].
\end{equation}
We assume that, in the limit, the following moment matrices tend to constant values when evaluated at $w^o_k$:
\begin{align}
\label{eq: constant covariance matrix}
R_{s,k}&\triangleq\lim_{i\rightarrow\infty}\expec\left[\bs_{k,i}(w^o_k)\bs_{k,i}^*(w^o_k)|\bcF_{i-1}\right],\\
\label{eq: constant covariance matrix 2}
R_{q,k}&\triangleq\lim_{i\rightarrow\infty}\expec\left[\bs_{k,i}(w^o_k)\bs_{k,i}^\top(w^o_k)|\bcF_{i-1}\right].
\end{align}
\begin{assumption}{{\emph{(Smoothness condition on noise covariance).}}}
\label{assumption: gradient noise moment}
It is assumed that the conditional second-order moments of the individual noise processes satisfy the following smoothness condition,
\begin{equation}
\|R_{s,k,i}^e(w^o_k+\Delta w_k)-R_{s,k,i}^e(w^o_k)\|\leq\kappa_d\|\Delta w_k\|^\gamma,\label{eq: smoothness condition on the covariance noise}
\end{equation}
for small perturbations $\|\Delta w_k\|$, and for some constant $\kappa_d\geq 0$ and exponent $0<\gamma\leq 4$. 
\end{assumption}
\noindent One useful conclusion that follows from~\eqref{eq: smoothness condition on the covariance noise} is that, for $i\gg 1$ and for sufficiently small step-size, we can express the covariance matrix of $\bs_{k,i}^e(\bw)$ in terms of the limiting matrix $R^o_k$ defined in Table~\ref{table: variables table} as follows (\cred{see~\cite[Lemma~11.1]{sayed2014adaptation}}): 
\begin{equation}
\expec\bs_{k,i}^e(\bw_{k,i-1})\bs_{k,i}^{e*}(\bw_{k,i-1})=R_{k}^o+O(\mu^{\min\{1,\frac{\gamma}{2}\}}).\label{eq: zero bound on the difference of covariances}
\end{equation}

Before proceeding, we introduce the $(hM)^2\times (hM)^2$ matrix $\cF$ that will play a critical role in characterizing the performance:
\begin{equation}
\cF=\cB^\top\otimes_b\cB^*.\label{eq: definition of the matrix cF}
\end{equation}
\cblue{This matrix is defined in in terms of the block Kronecker operation. In the derivation that follows, we shall use the block Kronecker product $\otimes_b$ operator~\cite{koning1991block} and the block vectorization operator $\bvc(\cdot)$\footnote{In our derivations, the block Kronecker product and the block vectorization operations are applied to $2\times 2$ block matrices $\cC=[C_{k\ell}]$ and $\cD=[D_{k\ell}]$ with blocks $\{C_{11},D_{11}\}$ of size $hP\times hP$, blocks $\{C_{12},D_{12}\}$ of size $hP\times h(M-P)$, blocks $\{C_{21},D_{21}\}$ of size $h(M-P)\times hP$, and blocks $\{C_{22},D_{22}\}$ of size $h(M-P)\times h(M-P)$.}. As explained in~\cite{sayed2014adaptation}, these operations preserve the locality of the blocks in the original matrix arguments.} The matrix $\cF$ will sometimes appear transformed under the similarity transformation:
\begin{equation}
\label{eq: definition cFb}
\cFb\triangleq\left((\cV^e_{\epsilon})^\top\otimes_b(\cV^e_{\epsilon})^*\right)\cF\left((\cV^e_{\epsilon})^\top\otimes_b(\cV^e_{\epsilon})^*\right)^{-1}.
\end{equation}
\begin{lemma}{\emph{(Low-rank approximation).}}
Assume the matrix $\cA$ satisfies conditions~\eqref{eq: first condition on eigenvector},~\eqref{eq: second condition on eigenvector}, and~\eqref{eq: third condition in the proposition} with $\cU$ satisfying Assumption~\ref{assump: matrix cU}. For sufficiently small step-sizes, it holds that $\cF$ is stable and that:
\begin{align}
(I-\cF)^{-1}&=O(1/\mu),\label{eq: order of inverse of I-cF}\\
(I-\cFb)^{-1}&=\left[\begin{array}{c|c}
O(1/\mu)&O(1)\\
\hline
O(1)&O(1)
\end{array}
\right],\label{eq: order of inverse of I-cFb}
\end{align}
where the leading $(hP)^2\times (hP)^2$ block in $(I-\cFb)^{-1}$ is $O(1/\mu)$. Moreover, we can also write:
\begin{equation}
\label{eq: low rank approximation of I-F inverse}
(I-\cF)^{-1}={\left([(\cU^e)^*]^\top\otimes_b\cU^e\right)Z^{-1}\left((\cU^e)^\top\otimes_b(\cU^e)^*\right)}+O(1),
\end{equation}
in terms of the block Kronecker operation, where the matrix $Z$ has dimension $(hP)^2\times (hP)^2$:
\begin{equation}
\label{eq: definition of the matrix Z}
Z=(I_{hP}\otimes\cD_{11}^*)+(\cD_{11}^\top\otimes I_{hP})=O(\mu),
\end{equation}
with $\cD_{11}=\mu\,(\cU^e)^*\cH^o\cU^e$ which is positive definite under Assumption~\ref{assump: risks}.
\end{lemma}
\begin{proof}
See Appendix~\ref{app: low-rank approximation}.
\end{proof}
\begin{theorem}{{\emph{(Mean-square-error performance).}}} 
\label{theo: Network limiting performance}
Consider the same settings of Theorem~\ref{theo: performance error}. Assume further that Assumption~\ref{assumption: gradient noise moment} holds.  Let $\gamma_m\triangleq\frac{1}{2}\min\{1,\gamma\}>0$ with $\gamma\in(0,4]$ from~\eqref{eq: smoothness condition on the covariance noise}. Then, it holds that:
\begin{align}
&\limsup_{i\rightarrow\infty}\frac{1}{{h}N}\expec\|\bcwt_{i}^e\|^2\nonumber\\
&=\frac{1}{hN}(\emph{\bvc}(\cY^\top))^\top(I-\cF)^{-1}\emph{\bvc}(I_{hM})+O(\mu^{1+\gamma_m}),\label{eq: network error variance}\\
&=\frac{1}{hN}\emph{\tr}\left(\sum_{n=0}^{\infty}\cB^n\cY(\cB^*)^n\right)+O(\mu^{1+\gamma_m}),\label{eq: network error variance 100}
\end{align}
where:
\begin{align}
\cY&=\mu^2\cA^e\cS(\cA^e)^*,\label{eq: definition of the matrix cY}\\
\cS&=\emph{\diag}\{R_{1}^o,R_{2}^o,\ldots,R_{N}^o\}.\label{eq: definition of S general}
\end{align}
 Furthermore, it holds that:
 \begin{equation}
\label{eq: final result for MSD}
\emph{MSD}=\frac{\mu}{2hN}\emph{\tr}\left(\left((\cU^e)^*\cH^o\cU^e\right)^{-1}\left((\cU^e)^*\cS\cU^e\right)\right).
\end{equation}
 \end{theorem}
\begin{proof}
See Appendix~\ref{app: proof of theorem 3}.
\end{proof}
Since  $(I-\cF)$ is of size $(hM)^2\times (hM)^2$, the first term on the R.H.S. of expression~\eqref{eq: network error variance} may be hard to evaluate due to numerical reasons. In comparison, the first term in expression~\eqref{eq: network error variance 100} only requires manipulations of matrices of size $hM\times hM$. In practice, a reasonable number of terms can be used instead of $n\rightarrow \infty$ to obtain accurate evaluation.

Note that the MSD of the centralized solution is equal to~\eqref{eq: final result for MSD} since the centralized implementation can be obtained from~\eqref{eq: distributed solution} by replacing $\cP_{\cu}$ by $\cA$ and by assuming fully-connected network. We therefore conclude, for sufficiently small step-sizes (i.e., in the slow adaptation regime),  that the distributed strategy~\eqref{eq: distributed solution} is able to attain the same MSD performance as the centralized solution. 

\section{Finding a combination matrix $\cA$}
\label{sec: finding a combination matrix A}
In the following, we consider the problem of finding an  $\cA$ that minimizes the number of edges to be added to the original topology while satisfying the constraints~\eqref{eq: third condition in the proposition},~\eqref{eq: first condition on eigenvector}, and~\eqref{eq: second condition on eigenvector}. That is, we consider the following optimization problem: 
\begin{equation}
\label{eq: optimization problem A}
\begin{array}{cl}
\minimize\limits_{\cA}&f(\cA)=\sum\limits_{k=1}^N\sum\limits_{\ell\notin\cN_k}\vertiii{A_{k\ell}}_1+\frac{\gamma}{2}\|\cA\|^2_{\text{F}},\\
\st&\cA\,\cU=\cU,
~~\cA=\cA^*,\\ 
&\rho(\cA-\cP_{\cu})\leq 1-\epsilon,
\end{array}
\end{equation} 
where $\vertiii{A_{k\ell}}_1\triangleq\sum_{m=1}^{M_k}\sum_{n=1}^{M_\ell}|[A_{k\ell}]_{mn}|\in\mathbb{R}$, $\|\cA\|_{\text{F}}=\sqrt{\tr(\cA^*\cA)}\in\mathbb{R}$ is the Frobenius norm of $\cA$, $\gamma\geq 0$ is a regularization parameter, and \cblue{$\epsilon\in(0,1]$} is a small positive number. In general, the spectral radius of a matrix is not convex \cblue{over the matrix space}. We therefore restrict our search to the class of Hermitian matrices, since their spectral radius coincides with their spectral norm (maximum singular value), which is a convex function. Problem~\eqref{eq: optimization problem A} is convex since the objective is convex, the equality constraints are linear, and the inequality constraint function is convex~\cite{boyd2004convex}. The parameter $\epsilon$ controls the convergence rate of $\cA^i$ towards the projector $\cP_{\cu}$. That is, small $\epsilon$ leads to slow convergence and large $\epsilon$ gives fast convergence. The convex $\ell_1$-norm based function $\sum_{k=1}^N\sum_{\ell\notin\cN_k}\vertiii{A_{k\ell}}_1$ is used as a  relaxation of the pseudo $\ell_0$-norm $h(\cA)=\sum_{k=1}^N\card\{\ell |A_{k\ell}\neq 0,\ell\notin\cN_k\}$, which is a non-convex function that leads to computational challenges. Among the potentially multiple feasible solutions, the cardinality function $h(\cA)$ in the objective in~\eqref{eq: optimization problem A} selects as optimum the one that minimizes the number of edges to be added to the network topology in order to satisfy constraint~\eqref{eq: condition 1 on A}. The quadratic term $\|\cA\|^2_{\text{F}}=\sum_{m=1}^M\sum_{n=1}^M|a_{mn}|^2$ in~\eqref{eq: optimization problem A} makes the objective strictly convex, and therefore problem~\eqref{eq: optimization problem A} has a unique minimum. Problem~\eqref{eq: optimization problem A} can be solved using general convex optimization solvers such as CVX~\cite{grant2014cvx}. These solvers generally implement second-order methods that require calculation of Hessian \cblue{matrices}. Therefore, problems with more than few thousand entries are probably beyond the capabilities of these solvers. The Douglas-Rachford algorithm can also be employed to solve problem~\eqref{eq: optimization problem A}. As we shall see in the following, the required proximal operators for implementing this algorithm can be computed efficiently using closed form expressions. 

In the following, we shall assume that $\cU\in\mathbb{R}^{M\times P}$ and, therefore, we shall solve~\eqref{eq: optimization problem A} over real-valued matrices $\cA\in\mathbb{R}^{M\times M}$. In order to solve the constrained problem~\eqref{eq: optimization problem A}, we shall apply the Douglas-Rachford splitting algorithm~\cite{combettes2011proximal}, which is used to solve problems of the form:
\begin{equation}
\minimize_{x\in\mathbb{R}^N}g_1(x)+g_2(x),
\end{equation}
where $g_1(\cdot)$ and $g_2(\cdot)$ are functions in $\Gamma_0(\mathbb{R}^N)$ such that $(\text{ri dom}g_1)\cap(\text{ri dom}g_2)\neq 0$ and $g_1(x)+g_2(x)\rightarrow+\infty$ as $\|x\|\rightarrow+\infty$. 
By selecting $g_1(\cdot)$ as $f(\cdot)$ in~\eqref{eq: optimization problem A} and $g_2(\cdot)$ as the indicator function $\cI_{\Omega}(\cdot)$ of the closed nonempty convex set:
\begin{equation}
\label{eq: set Omega}
\Omega=\{\cA|\cA\,\cU=\cU,\cA=\cA^\top,\|\cA-\cP_{\cu}\|\leq 1-\epsilon\}
\end{equation}
 defined as:
\begin{equation}
\cI_{\Omega}(\cA)\triangleq\left\lbrace
\begin{array}{ll}
0,& \text{if }\cA\in\Omega,\\
+\infty,&\text{if }\cA\notin\Omega,
\end{array}
\right.
\end{equation}
the Douglas-Rachford algorithm to solve~\eqref{eq: optimization problem A} has the following form:
\begin{equation}
\label{eq: douglas rachford algorithm}
\left\lbrace
\begin{split}
\cA_i&=\prox_{\eta f}(\cC_i)\\
\cC_{i+1}&=\cC_{i}+\prox_{\eta \cI_\Omega}(2\cA_i-\cC_i)-\cA_i,
\end{split}
\right.
\end{equation}
where $\eta>0$ and $\prox_{\eta g}:\mathbb{R}^{M\times M}\rightarrow\mathbb{R}^{M\times M}$ is the proximal operator of $\eta g(\cdot)$ ($g:\mathbb{R}^{M\times M}\rightarrow\mathbb{R}\cup\{+\infty\}$) defined as~\cite{combettes2011proximal,parikh2014proximal}:
\begin{equation}
\prox_{\eta g}(\cC)=\arg\min_{\cA}g(\cA)+\frac{1}{2\eta}\|\cA-\cC\|^2_\text{F}.
\end{equation}
Every sequence $(\cA_i)_{i\in\mathbb{N}}$ generated by algorithm~\eqref{eq: douglas rachford algorithm} converges to a solution of problem~\eqref{eq: optimization problem A}~\cite[Proposition 4.3]{combettes2011proximal}. The Douglas-Rachford algorithm operates by splitting since it employs the functions $f(\cdot)$ and $\cI_\Omega(\cdot)$ separately. It requires the implementation of two proximal steps at each iteration\cblue{,} which can be computed efficiently as explained in the following. 

The function $f(\cA)$ is an entrywise matrix function that treats the matrix $\cA\in\mathbb{R}^{M\times M}$ as a vector in $\mathbb{R}^{M^2}$ and then uses a corresponding vector function; the proximal operator is then the same as that of the vector function. Let $C_{k\ell}$ denote the $(k,\ell)$-th block of an $N\times N$ block matrix $\cC$ and let $[C_{k\ell}]_{mn}$ denote the $(m,n)$-th entry of $C_{k\ell}$. The $(k,\ell)$-th block of the proximal operator of $\eta f(\cdot)$ is given by:
\begin{equation}
[\prox_{\eta f}(\cC)]_{k\ell}=\left(\frac{1}{1+\eta\gamma}\right)\cdot\left\lbrace
\begin{array}{ll}
C_{k\ell},& \text{if } \ell\in\cN_k \text{ or } k=\ell,\\
C^s_{k\ell},& \text{if } \ell\notin\cN_k ,
\end{array}
\right.
\end{equation}
where the matrix $C^s_{k\ell}$ is of size $M_k\times M_\ell$ with $(m,n)$-th entry given by:
\begin{equation}
{[C^s_{k\ell}]}_{mn}=\left\lbrace
\begin{array}{ll}
{[C_{k\ell}]}_{mn}-\eta,& \text{if }  [C_{k\ell}]_{mn}\geq\eta,\\
0,& \text{if } |[C_{k\ell}]_{mn}|\leq\eta,\\
{[C_{k\ell}]}_{mn}+\eta,& \text{if }  {[C_{k\ell}]}_{mn}\leq-\eta.
\end{array}
\right.
\end{equation}
Since $\cI_\Omega$ is the indicator function of the closed convex set $\Omega$, its proximal operator reduces to the projection onto $\Omega$ defined as:
\begin{equation}
\prox_{\eta\cI_{\Omega}}(\cD)=\Pi_{\Omega}(\cD)=
\left\lbrace
\begin{array}{cl}
\arg\min\limits_\cA&\frac{1}{2}\|\cA-\cD\|^2_{\text{F}}\\
\st&\cA\in\Omega
\end{array}
\right.
\end{equation}
where  the parameter $\eta$ does not appear since the proximal operator is a projection. The set $\Omega$ in~\eqref{eq: set Omega} can be written alternatively as $\Omega=\Omega_1\cap\Omega_2$ where $\Omega_1$ and $\Omega_2$ are two closed convex sets defined as:
\begin{align}
\Omega_1&=\{\cA|\cA\,\cU=\cU,\cA=\cA^\top\},\label{eq: set Omega 1}\\
\Omega_2&=\{\cA|\|\cA-\cP_{\cu}\|\leq 1-\epsilon\}.\label{eq: set Omega 2}
\end{align}
As we shall explain in the following, the projection onto the intersection $\Omega$ can be obtained by properly projecting onto the individual sets $\Omega_1$ and $\Omega_2$ according to:
\begin{equation}
\vspace{-1.5mm}
\label{eq: projection onto omega the intersection}
\Pi_{\Omega}(\cD)=\Pi_{\Omega_2}(\Pi_{\Omega_1}(\cD)).
\end{equation}
The projection onto $\Omega_1$ is given by (see Appendix~\ref{app: projection onto Omega 1}):
\begin{equation}
\vspace{-1.5mm}
\Pi_{\Omega_1}(\cD)=(I-\cP_{\cu})\left(\frac{\cD+\cD^\top}{2}\right)(I-\cP_{\cu})+\cP_{\cu},
\end{equation}
and the projection of the symmetric matrix $\Pi_{\Omega_1}(\cD)$ onto $\Omega_2$ is given by (see Appendix~\ref{app: projection onto Omega 1}):
\begin{equation}
\label{eq: projection onto omega 2}
\Pi_{\Omega_2}(\Pi_{\Omega_1}(\cD))=\cP_{\cu}+\sum_{m=1}^M\beta_mv_mv_m^\top,
\end{equation}
where:
\begin{equation}
\beta_m=\left\lbrace
\begin{array}{ll}
-1+\epsilon, &\text{if } \lambda_m<-1+\epsilon,\\
\lambda_m, &\text{if }  |\lambda_m|<1-\epsilon,\\
1-\epsilon, &\text{if }  \lambda_m>1-\epsilon.
\end{array}
\right.
\end{equation}
where $\{\lambda_m,v_m\}_{m=1}^M$ are the eigenvalues and eigenvectors of the symmetric matrix $(I-\cP_{\cu})\left(\frac{\cD+\cD^\top}{2}\right)(I-\cP_{\cu})$. In order to establish~\eqref{eq: projection onto omega the intersection}, we introduce the following lemma.
\begin{lemma}{\emph{(Characterization of the projection).}}
\label{lemma: one step projection}
If $\Omega_1$ is an affine set, $\Omega_2$ is a closed convex set, and $\Pi_{\Omega_2}(\Pi_{\Omega_1}(\cC))\in\Omega_1$, then $\Pi_{{\Omega_1}\cap{\Omega_2}}(\cC)=\Pi_{\Omega_2}(\Pi_{\Omega_1}(\cC))$.
\end{lemma}
\begin{proof}
See Appendix~\ref{app: proof of lemma 4}.
\end{proof}
Since the projection onto $\Omega_2$ (given by~\eqref{eq: projection onto omega 2}) changes only the eigenvalues of a matrix without affecting the eigenvectors, we have $\Pi_{\Omega_2}(\Pi_{\Omega_1}(\cC))\in\Omega_1$. We then conclude from Lemma~\ref{lemma: one step projection} that~\eqref{eq: projection onto omega the intersection} holds.

\section{Simulation results}
\label{subsec: theoretical model validation}
\begin{figure}
\begin{center}
\includegraphics[scale=0.4]{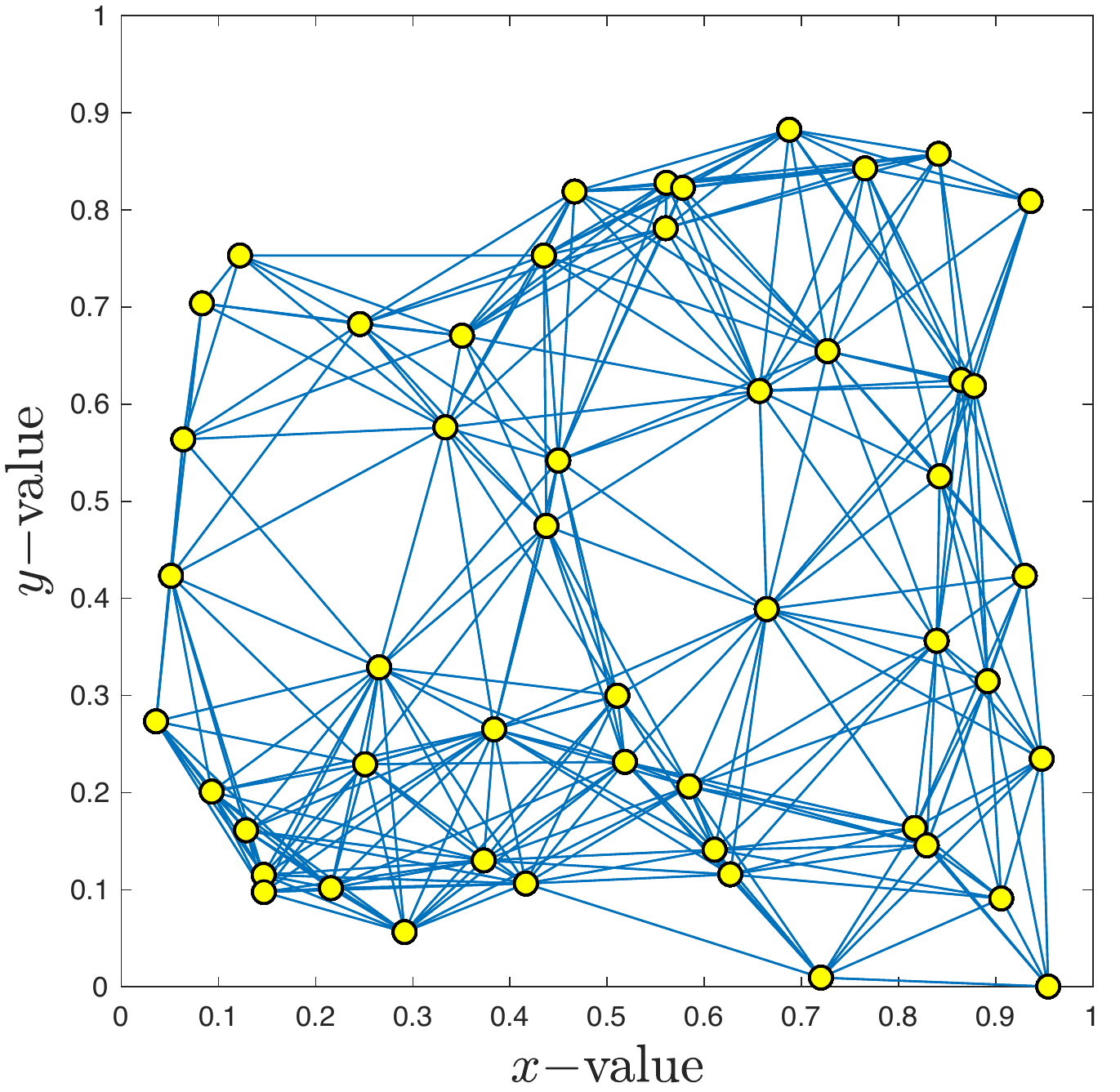}\qquad
\includegraphics[scale=0.4]{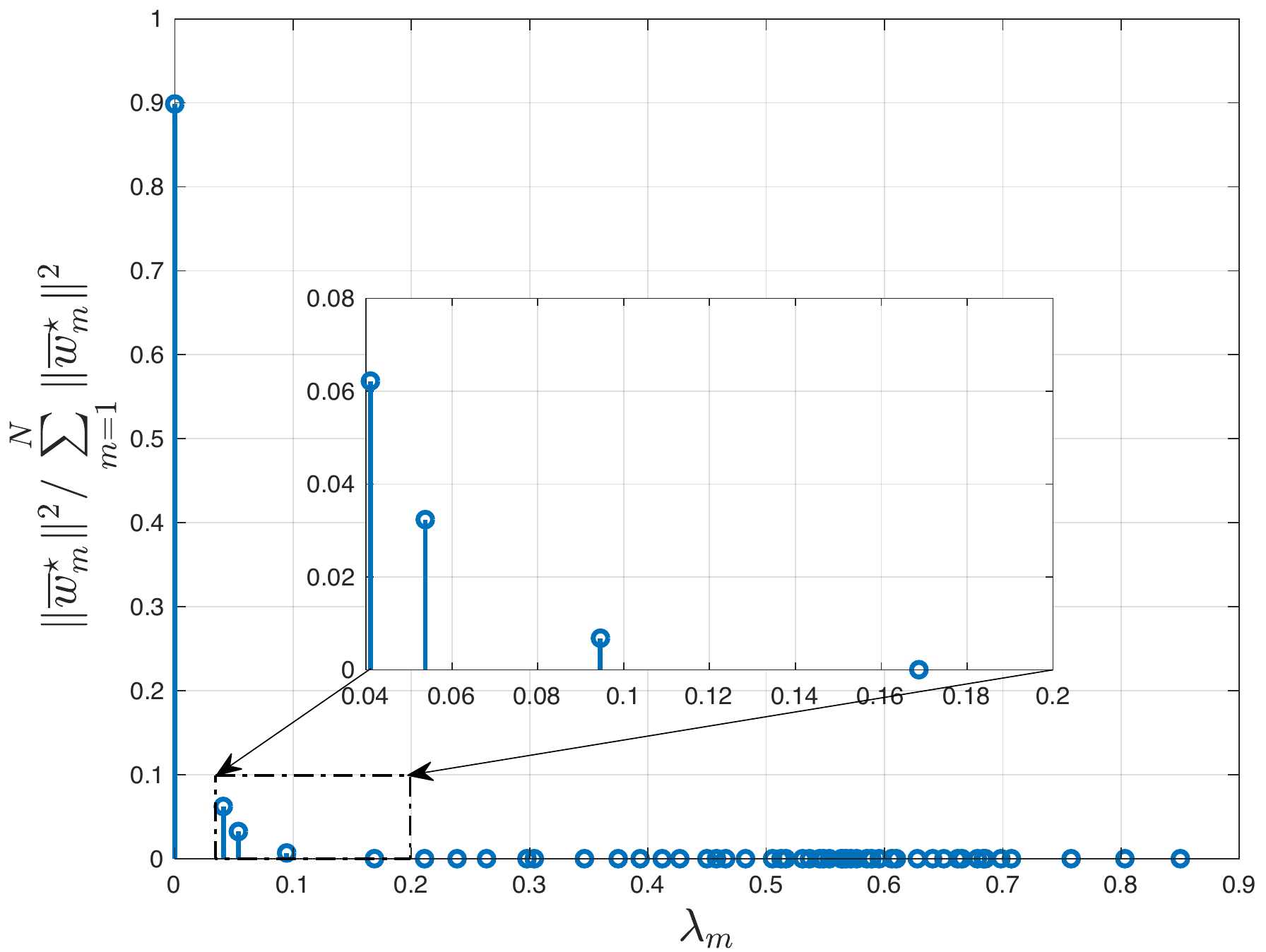}
\caption{Experimental setup. \textit{(Left)} Network topology. \textit{(Right)} Graph spectral content of $\cw^\star$ with $\wb^\star_m=(v_m^\top\otimes I_L)\cw^\star$.}
\label{fig: data settings}
\end{center}
\vspace{-5mm}
\end{figure}
We apply strategy~\eqref{eq: distributed solution} to solve distributed inference under smoothness (described in Remark 4 of Section~II in~\cite{nassif2019adaptation}). We consider a connected mean-square-error (MSE) network of $N=50$ nodes and $M_k=L=5$. The $N$ nodes are placed randomly in the $[0,1]\times[0,1]$ square, and the weighted graph is then constructed according to a thresholded Gaussian kernel weighting function based on the distance between nodes. Particularly, the weight $c_{k\ell}$ of edge $(k,\ell)$ connecting nodes $k$ and $\ell$ that are $d_{k\ell}$ apart is:
\begin{equation}
\label{eq: simulation adjacency matrix}
c_{k\ell}=\left\lbrace
\begin{array}{ll}
\exp{\left(-d_{k\ell}^2/(2\sigma^2)\right)}, &\text{if }d_{k\ell}\leq\kappa\\
0,&\text{otherwise}
\end{array}
\right.
\end{equation}
with $\sigma=0.12$ and $\kappa=0.33$. We assume real data case.  Each agent is subjected to streaming data $\{\bd_k(i),\bu_{k,i}\}$ assumed to satisfy a linear regression model~\cite{sayed2014adaptation}:
\begin{equation}
\bd_k(i)=\bu_{k,i}^\top w^\star_k+\bv_{k}(i),\quad k=1,\ldots,N,
\end{equation}
for some unknown $L\times 1$ vector $w^\star_k$ to be estimated with $\bv_k(i)$ denoting a zero-mean measurement noise. For these networks, the risk functions take the form of mean-square-error costs:
\begin{equation}
J_k(w_k)=\frac{1}{2}\expec|\bd_k(i)-\bu_{k,i}^\top w_k|^2,\quad k=1,\ldots,N.
\end{equation}
The processes $\{\bu_{k,i},\bv_{k}(i)\}$ are assumed to be zero-mean Gaussian with: i) $\expec\bu_{k,i}\bu_{\ell,i}^\top=R_{u,k}=\sigma^2_{u,k}I_L$ if $k=\ell$ and zero otherwise; ii) $\expec\bv_{k}(i)\bv_{\ell}(i)=\sigma^2_{v,k}$ if $k=\ell$ and zero otherwise; and iii) $\bu_{k,i}$ and $\bv_k(i)$ are independent of each other. The variances $\sigma^2_{u,k}$ and $\sigma^2_{v,k}$ are generated from the uniform distributions $\text{unif}(0.5,2)$ and $\text{unif}(0.2,0.8)$, respectively. Let $\cw^\star=\col\{w_1^\star,\ldots,w^\star_N\}$. The signal $\cw^\star$ is generated by smoothing a signal $\cw_o$ by a diffusion kernel. Particularly, we generate $\cw^\star$ according to $\cw^\star=[(Ve^{-\tau\Lambda}V^\top)\otimes I_L]\cw_o$ with $\tau=30$, $\cw_o$ a randomly generated vector from the Gaussian distribution $\cN(0.1\times\mathds{1}_{NL},I_{NL})$, and $\{V=[v_1,\ldots,v_N],\Lambda=\diag\{\lambda_1,\ldots,\lambda_N\}\}$ are the matrices of eigenvectors and eigenvalues of $L_c=\diag\{C\mathds{1}_N\}-C$ with $[C]_{k\ell}=c_{k\ell}$ given by~\eqref{eq: simulation adjacency matrix}. Figure~\ref{fig: data settings} (right) illustrates the normalized squared $\ell_2$-norm of the spectral component $\wb^\star_m=(v_m^\top\otimes I_L)\cw^\star$. It can be observed that the signal is mainly localized in $[0,0.1]$. Note that, for MSE networks, it holds that $H_k(w_k)=R_{u,k}$ $\forall w_k$. Furthermore, the gradient noise process~\eqref{eq: gradient noise process} is given by:
\begin{equation}
\label{eq: MSE network gradient noise}
\bs_{k,i}(\bw_{k})=(\bu_{k,i}^{\top}\bu_{k,i}-R_{u,k})(w^o_k-\bw_{k})+\bu_{k,i}^\top\bv_{k}(i),
\end{equation}
with covariance $R^o_{k}$ given by:
\begin{align}
R^o_{k}&=\expec[(\bu_{k,i}^{\top}\bu_{k,i}-R_{u,k})W_{k}(\bu_{k,i}^{\top}\bu_{k,i}-R_{u,k})]+\sigma^2_{v,k}R_{u,k}\notag\\
&=R_{u,k}W_{k}R_{u,k}+R_{u,k}\tr(R_{u,k}W_{k})+\sigma^2_{v,k}R_{u,k}\label{eq: covariance noise MSE 1}
\end{align}
where $W_{k}=(w^\star_k-w^o_{k})(w^\star_k-w^o_{k})^\top$, and where we used the fact that $\expec[\bu_{k,i}^{\top}\bu_{k,i}W_{k}\bu_{k,i}^{\top}\bu_{k,i}]=2R_{u,k}W_{k}R_{u,k}+R_{u,k}\tr(R_{u,k}W_{k})$ since the regressors are zero-mean real Gaussian~\cite{isserlis1918formula}.

\begin{figure*}
\begin{center}
\includegraphics[scale=0.4]{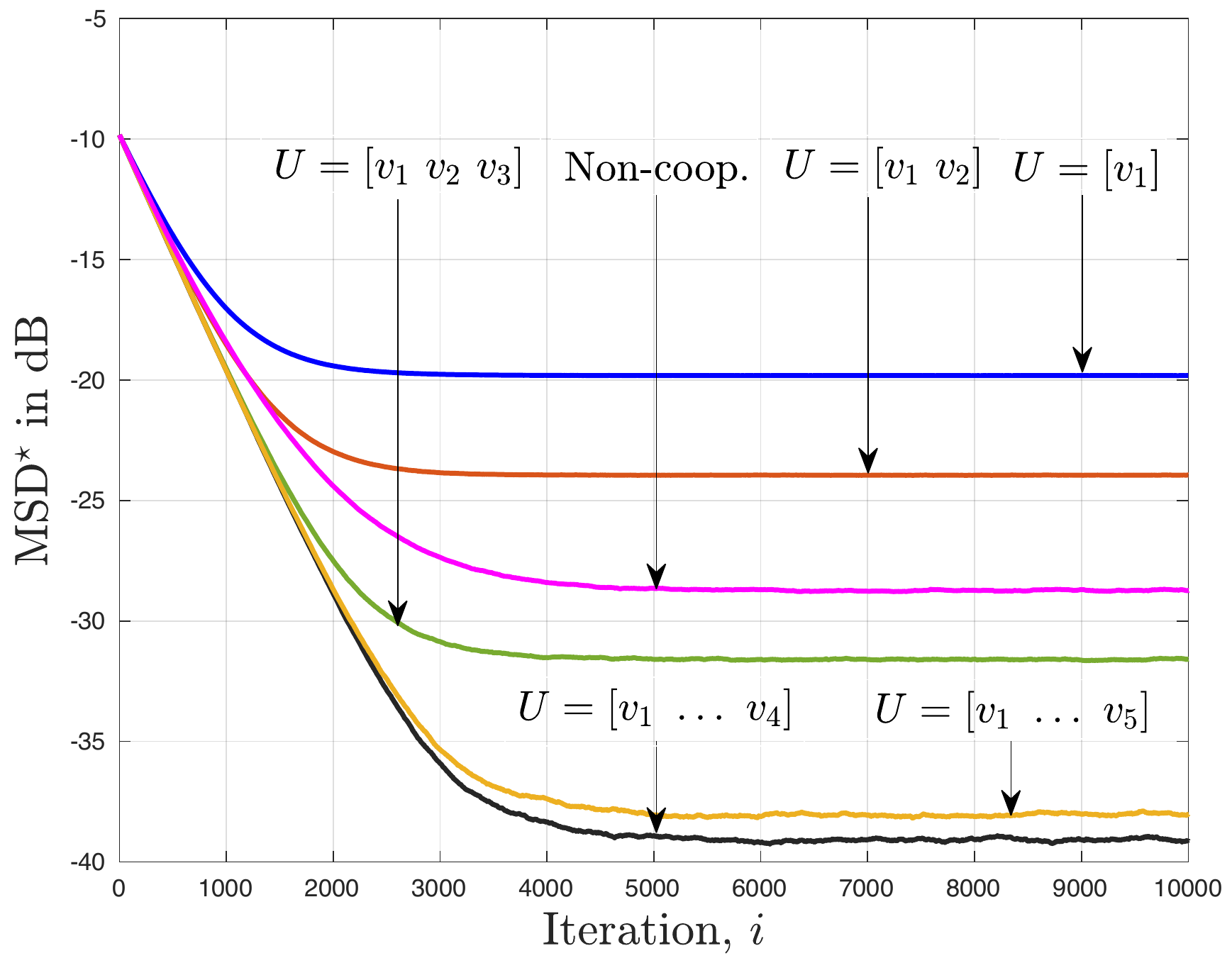}\qquad
\includegraphics[scale=0.4]{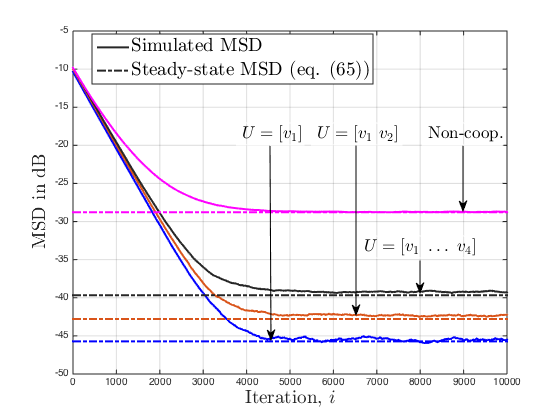}
\caption{Inference under smoothness. Performance of algorithm~\eqref{eq: distributed solution} for $5$ different choices of the matrix $\cU$ in~\eqref{eq: constrained optimization problem} with $\cU=U\otimes I_L$, and non-cooperative strategy. \textit{(Left)} Performance w.r.t. $\cw^\star$. \textit{(Right)} Performance w.r.t. $\cw^o$.}
\label{fig: network performance}
\end{center}
\vspace{-6mm}
\end{figure*}

We run algorithm~\eqref{eq: distributed solution} for 5 different choices of matrix $\cU$ in~\eqref{eq: constrained optimization problem} with $\cU=U\otimes I_L$: i) matrix $U$ chosen as the first eigenvector of the Laplacian $U=[v_1]=\frac{1}{\sqrt{N}}\mathds{1}_N$; ii) matrix $U$ chosen as the first two eigenvectors of the Laplacian $U=[v_1~v_2]$;  iii) $U=[v_1~v_2~v_3]$; iv) $U=[v_1~\ldots~v_4]$; v) $U=[v_1~\ldots~v_5]$.  Since $\cU=U\otimes I_L$, the matrix $\cA$ is of the form $A\otimes I_L$ with $A=[a_{k\ell}]$ an $N\times N$ matrix. In each case, the combination matrix $A$ is set as the solution of the optimization problem~\eqref{eq: optimization problem A} ($\epsilon=0.01,\gamma=0$, $\vertiii{A_{k\ell}}_1=|a_{k\ell}|$), which is solved by the Douglas-Rachford algorithm~\eqref{eq: douglas rachford algorithm} with $\eta=0.003$. Note that, for the $5$ different choices of $U$, the distributed implementation is feasible and the steady-state value of the cost in~\eqref{eq: optimization problem A} is zero. We set $\mu=0.001$. We report the network $\text{MSD}^\star$ learning curves $\frac{1}{N}\expec\|\cw^\star-\bcw_i\|^2$ in Fig.~\ref{fig: network performance} (left). The results are averaged over $200$ Monte-Carlo runs. The learning curve of the non-cooperative solution, obtained from~\eqref{eq: distributed solution} by setting $\cA=I_{LN}$, is also reported. The results show that the best performance is obtained when $\cU=[v_1~v_2~v_3~v_4]\otimes I_L$. This is due to the fact that the columns of $\cU$ constitute a basis spanning the useful signal subspace (see Fig.~\ref{fig: data settings} (right)). As a consequence, a strong noise reduction may be obtained by projecting onto this subspace compared with the non-cooperative strategy where each agent estimates $w^\star_k$ without any cooperation. By forcing consensus (i.e., by choosing $U=[v_1]$), the resulting estimate $\bw_{k,i}$ will be biased with respect to $w^\star_k$, which is not common across agents. The performance obtained when $U=[v_1~\ldots~v_5]$ is worse than the case where $U=[v_1~\ldots~v_4]$ due to a smaller noise reduction.

\begin{table}
\caption{Performance of strategy~\eqref{eq: distributed solution} w.r.t. $\cw^o$ in~\eqref{eq: constrained optimization problem} for $2$ different choices of $\mu$.}
\begin{center}
\begin{tabular}{|c||c||c|c|c|}
\hline\hline
\cellcolor{gray!60}& \cellcolor{gray!60}&  \multicolumn{3}{|c|}{ \cellcolor{gray!60}MSD}  \\ 
\cline{3-5}
\multirow{ -2}{*}{\cellcolor{gray!60}Step-size $\mu$}&\multirow{ -2}{*}{\cellcolor{gray!60}Solution} &  \cellcolor{gray!15}Exp.~\eqref{eq: final result for MSD}&  \cellcolor{gray!15}Exp.~\eqref{eq: network error variance 100}& \cellcolor{gray!15}Simulation  \\ \hline\hline
\cellcolor{gray!60}&\cellcolor{gray!15} Centralized &$-29.66$dB &$-29.74$dB&$-29.603$dB \\
\cline{2-5}
\multirow{ -2}{*}{\cellcolor{gray!60}$10^{-2}$} &\cellcolor{gray!15}Distributed&$-29.66$dB&$-27.6$dB&$-27.298$dB\\ \hline
\cellcolor{gray!60}& \cellcolor{gray!15}Centralized &$-39.66$dB &$-39.67$dB&$-39.691$dB \\
\cline{2-5}
\multirow{- 2}{*}{\cellcolor{gray!60}$10^{-3}$} &\cellcolor{gray!15}Distributed&$-39.66$dB&$-39.26$dB&$-39.196$dB \\ \hline
\cellcolor{gray!60} & \cellcolor{gray!15}Centralized &$-49.66$dB &$-50.19$dB&$-49.772$dB \\
\cline{2-5}
\multirow{- 2}{*}{\cellcolor{gray!60}$10^{-4}$} &\cellcolor{gray!15}Distributed&$-49.66$dB&$-50.14$dB&$-49.335$dB\\ \hline\hline
\end{tabular}
\end{center}
\label{table: MSD}
\vspace{-5mm}
\end{table}
Finally, we illustrate Theorem~\ref{theo: Network limiting performance} in Table~\ref{table: MSD} by reporting the steady-state $\text{MSD}=\limsup_{i\rightarrow\infty}\frac{1}{N}\expec\|\cw^o-\bcw_i\|^2$ when $\cU=[v_1~\ldots~v_4]\otimes I_L$ for $3$ different values of the step-size $\mu=\{10^{-2},10^{-3},10^{-4}\}$. A closed form solution for $\cw^o$ in~\eqref{eq: constrained optimization problem} exists and is given by:
\begin{equation}
\cw^o=\cU(\cU^\top\cH\,\cU)^{-1}\cU^\top\cH\cw^\star,
\end{equation}
where $\cH=\diag\{R_{u,k}\}_{k=1}^N$. We observe that, in the small adaptation regime, i.e., when $\mu\rightarrow 0$, the network $\text{MSD}$ increases approximately $10$dB per decade (when $\mu$ goes from $\mu_1$ to $10\mu_1$). This means that the steady-state $\text{MSD}$ is on the order of $\mu$. We also observe that, in the small adaptation regime, the distributed solution is able to attain the same performance as the centralized one. Finally, note that, for relatively large step-size ($\mu=10^{-2}$), expression~\eqref{eq: network error variance 100} provides better results than~\eqref{eq: final result for MSD} in the distributed case. This is due to neglecting the $O(1)$ term in~\eqref{eq: order of inverse of I-cFb} which is multiplied by $O(\mu^2)$ (since $\cY=O(\mu^2)$) when replaced in~\eqref{eq: network error variance 100}.

 \section{Conclusion}
 \label{sec: conclusion}
In this paper, and its accompanying Part I~\cite{nassif2019adaptation}, we considered inference problems over networks where agents have individual parameter vectors to estimate subject to subspace constraints that require the parameters across the network to lie in low-dimensional subspaces. Based on the gradient projection algorithm, we proposed an iterative and distributed implementation of the projection step, which runs in parallel with the stochastic gradient descent update. We showed that, for small step-size parameter, the network is able to approach the minimizer of the constrained problem to arbitrarily good accuracy levels. Furthermore, we derived a closed-form expressions for the steady-state mean-square-error (MSE) performance. These expressions revealed explicitly the influence of the gradient noise, data characteristics, and subspace constraints, on the network performance. Finally, among many possible convex formulations, we considered the design of feasible distributed solution that minimizes the number of edges to be added to the original graph.

\begin{appendices}

\section{Proof of Theorem~\ref{theo: performance error}}
\label{app: Size of the approximation error}
To simplify the notation, we introduce the difference vector 
$\bz_i\triangleq\bcwt_i^e-\bcwt^{e'}_i.$
Using~\eqref{eq: definition of bcHt} in the expression for $\bcB_{i-1}$ in~\eqref{eq: matrix Bi-1 1}, we can write:
\begin{equation}
\label{eq: B i-1 and B}
\bcB_{i-1}=\cB+\mu\cA^e\bcHt_{i-1},
\end{equation}
in terms of the constant coefficient matrix $\cB$ in~\eqref{eq: constant B}.  Using~\eqref{eq: B i-1 and B} and~\eqref{eq: definition of the perturbation}, and subtracting~\eqref{eq: error recursion 1}  and~\eqref{eq: long term error model 2}, we then get:
\begin{equation}
\label{eq: error vector z}
\bz_i=\cB\bz_{i-1}+\mu\cA^e\bc_{i-1}.
\end{equation}
If we multiply both sides of~\eqref{eq: error vector z} from the left by $(\cV^e_{\epsilon})^{-1}$ we obtain:
\begin{equation}
\label{eq: error vector z 1}
\left[\begin{array}{c}
\overline{\bz}_{i}\\
\widecheck{\bz}_{i}
\end{array}\right]=\cBb\left[\begin{array}{c}
\overline{\bz}_{i-1}\\
\widecheck{\bz}_{i-1}
\end{array}\right]+\left[\begin{array}{c}
\overline{\bc}_{i-1}\\
\widecheck{\bc}_{i-1}
\end{array}\right]
\end{equation}
where we partitioned the vectors $(\cV^e_{\epsilon})^{-1}\bz_i$ and $\mu\Lambda_{\epsilon}^e(\cV^e_{\epsilon})^{-1}\bc_{i-1}$ into:
\begin{equation}
\label{eq: mean error recursion 6} 
(\cV^e_{\epsilon})^{-1}\bz_i\triangleq\left[\begin{array}{c}
\overline{\bz}_{i}\\
\widecheck{\bz}_{i}
\end{array}\right],\quad
\mu\Lambda_{\epsilon}^e(\cV^e_{\epsilon})^{-1}\bc_{i-1}\triangleq\left[\begin{array}{c}
\overline{\bc}_{i-1}\\
\widecheck{\bc}_{i-1}
\end{array}\right]
\end{equation}
with the leading vectors, $\{\overline{\bz}_{i},\overline{\bc}_{i-1}\}$, having dimensions $hP\times 1$ each. The matrix $\cBb$ is given by:
\begin{align}
\cBb&\triangleq(\cV_{\epsilon}^{e})^{-1}\cB\cV_{\epsilon}^e\notag\\
&\hspace{-4.5mm}\overset{\eqref{eq: constant B},\eqref{eq: jordan decomposition of A e}}{=}\Lambda^e_{\epsilon}-\mu\Lambda^e_{\epsilon}(\cV_{\epsilon}^{e})^{-1}\cH^o\cV_{\epsilon}^e\notag\\
&=\left[\begin{array}{cc}
I_{hP}-\cD_{11}&-\cD_{12}\\
-\cD_{21}&\cJ_{\epsilon}^e-\cD_{22}
\end{array}\right]\label{eq: cBb definition}
\end{align}
with the blocks $\{\cD_{mn}\}$ given by:
\begin{align}
\cD_{11}&\triangleq\mu\,(\cU^e)^*\cH^o\cU^e,\qquad\cD_{12}\triangleq\mu\,(\cU^e)^*\cH^o\cV_{R,\epsilon}^e,\label{eq: definition cD 11}\\
\cD_{21}&\triangleq\mu\cJ_{\epsilon}^e(\cV^e_{L,\epsilon})^*\cH^o\cU^e,~\cD_{22}\triangleq\mu\cJ_{\epsilon}^e(\cV^e_{L,\epsilon})^*\cH^o\cV_{R,\epsilon}^e.\label{eq: definition cD 22}
\end{align}
Recursion~\eqref{eq: error vector z 1}  has a form similar to the earlier recursion~(66) in Part~I~\cite{nassif2019adaptation} with three differences. First, the matrices $\{\cD_{mn}\}$ in~\eqref{eq: mean error recursion 6} are constant matrices; nevertheless, they satisfy the same bounds as the matrices $\{\bcD_{mn,i-1}\}$ in~eq.~(66) in Part~I~\cite{nassif2019adaptation}. In particular, from~(115),~(116), and~(122) in Part~I~\cite{nassif2019adaptation}, it continues to hold that:
\begin{align}
&\|I_{hP}-\cD_{11}\|\leq1-\mu\sigma_{11},\quad\|\cD_{12}\|\leq\mu\sigma_{12},\\
&\|\cD_{21}\|\leq\mu\sigma_{21},~~\quad\qquad\qquad\|\cD_{22}\|\leq\mu\sigma_{22},
\end{align}
for some positive constants $\sigma_{11},\sigma_{12},\sigma_{21},\sigma_{22}$ that are independent of $\mu$. Second, Third, the bias term $\widecheck{b}^e$ in~(66) in Part~I~\cite{nassif2019adaptation}  is absent from~\eqref{eq: mean error recursion 6}. Third, the gradient noise terms that appeared in recursion~(66) in Part~I~\cite{nassif2019adaptation} are now replaced by the perturbation sequences $\{\overline{\bc}_{i-1},\widecheck{\bc}_{i-1}\}$. However, these sequences can be bounded as follows:
\begin{equation}
\|\overline{\bc}_{i-1}\|^2\leq \mu^2r^2\|\bcwt^e_{i-1}\|^4,\quad\|\widecheck{\bc}_{i-1}\|^2\leq \mu^2r^2\|\bcwt^e_{i-1}\|^4,
\end{equation}
for some constant $r$ that is independent of $\mu$ since:
\begin{equation}
\label{equation 110}
\|\mu\Lambda_{\epsilon}^e(\cV^e_{\epsilon})^{-1}\bc_{i-1}\|^2\leq \mu^2r^2\|\bcwt^e_{i-1}\|^4. 
\end{equation}
To establish the above inequality, we start by noting that any cost $J_k(\cdot)$ satisfying~\eqref{eq: convexity condition} and~\eqref{eq: smoothness of the individual Hessian assumption} will also satisfy~\cite[Lemmas~E.4, E.8]{sayed2014adaptation}:
\begin{equation}
\label{eq: globally Lipshitz Hessian}
\|\nabla_{w_k}^2J_k(w^o_k+\Delta w_k)-\nabla_{w_k}^2J_k(w^o_k)\|\leq\kappa'_d\|\Delta w_k\|,
\end{equation}
for any $\Delta w_k$ and where $\kappa'_d\triangleq\max\{\kappa_d,(\delta_{k}-\nu_{k})/(h\cred{\epsilon})\}$. Then, for each agent $k$ we have:
\begin{align}
\label{eq: intermediate equation 1}
\|H_k^o-\bH_{k,i-1}\|&\overset{\eqref{eq: Hki-1},\eqref{eq: H star k}}{\leq}\int_{0}^1\|\nabla_{w_k}^2J_k(w^o_k)-\nabla_{w_k}^2J_k(w^o_k-t\bwt_{k,i-1})\|dt\notag\\
&\overset{\eqref{eq: globally Lipshitz Hessian}}{\leq}\int_{0}^1\kappa'_d\|t\bwt_{k,i-1}\|dt=\frac{1}{2}\kappa'_d\|\bwt_{k,i-1}\|
\end{align}
Therefore,
\begin{equation}
\label{eq: intermediate equation 3}
\begin{split}
\|\bcHt_{i-1}\|&\overset{\eqref{eq: definition of bcHt}}=\max_{1\leq k\leq N}\|H_k^o-\bH_{k,i-1}\|\\
&\leq\frac{1}{2}\kappa'_d\left(\max_{1\leq k\leq N}\|\bwt_{k,i-1}\|\right)\leq\frac{1}{2}\kappa'_d\|\bcwt_{i-1}^e\|.
\end{split}
\end{equation}
Now, replacing $\bc_{i-1}$ in~\eqref{equation 110} by~\eqref{eq: definition of the perturbation} and using~\eqref{eq: intermediate equation 3} we conclude~\eqref{equation 110}.

Repeating the argument that led to inequalities~(129) and (130) in Part~I~\cite{nassif2019adaptation} we obtain:
\begin{equation}
\label{eq: squared norm of b z c i}
\begin{split}
\expec\|\overline{\bz}_{i}\|^2\leq&(1-\mu\sigma_{11})\expec\|\overline{\bz}_{i-1}\|^2+\frac{2\mu\sigma^2_{12}}{\sigma_{11}}\expec\|\widecheck{\bz}_{i-1}\|^2+\frac{2\mu r^2}{\sigma_{11}}\expec\|\bcwt^e_{i-1}\|^4
\end{split}
\end{equation}
and
\begin{equation}
\label{eq: squared norm of b z r i}
\begin{split}
\expec\|\widecheck{\bz}_{i}\|^2\leq\left(\rho(\cJ_{\epsilon})+\epsilon+\frac{3\mu^2\sigma_{22}^2}{1-\rho(\cJ_{\epsilon})-\epsilon}\right)\expec\|\widecheck{\bz}_{i-1}\|^2+\frac{3\mu^2\sigma^2_{21}}{1-\rho(\cJ_{\epsilon})-\epsilon}\expec\|\overline{\bz}_{i-1}\|^2+\frac{3\mu^2 r^2}{1-\rho(\cJ_{\epsilon})-\epsilon}\expec\|\bcwt^e_{i-1}\|^4.
\end{split}
\end{equation}
We can combine~\eqref{eq: squared norm of b z c i} and~\eqref{eq: squared norm of b z r i} into a single inequality recursion as follows:
\begin{equation}
\left[\begin{array}{c}
\expec\|\overline{\bz}_{i}\|^2\\
\expec\|\widecheck{\bz}_{i}\|^2
\end{array}\right]\preceq\underbrace{\left[
\begin{array}{cc}
a&b\\
c&d
\end{array}
\right]}_{\Gamma}\left[\begin{array}{c}
\expec\|\overline{\bz}_{i-1}\|^2\\
\expec\|\widecheck{\bz}_{i-1}\|^2
\end{array}\right]+
\left[\begin{array}{c}
e\\
f
\end{array}\right]\expec\|\bcwt^e_{i-1}\|^4.
\end{equation}
where $a=1-O(\mu)$, $b=O(\mu)$,  $c=O(\mu^2)$, $d=\rho(\cJ_{\epsilon})+\epsilon+O(\mu^2)$, $e=O(\mu)$, and $f=O(\mu^2)$. 
Using~\eqref{eq: fourth order moment result} and eq.~(134) in Part~I~\cite{nassif2019adaptation} we conclude that:
\begin{equation}
\limsup_{i\rightarrow\infty}\expec\|\overline{\bz}_{i}\|^2=O(\mu^2),\quad\limsup_{i\rightarrow\infty}\expec\|\widecheck{\bz}_{i}\|^2=O(\mu^4),
\end{equation}
and, hence, 
$\limsup_{i\rightarrow\infty}\expec\|\bz_{i}\|^2=O(\mu^2).$
It follows that 
$\limsup_{i\rightarrow \infty}\expec\|\bcwt_i^e-\bcwt^{e'}_i\|^2=O(\mu^2),$ 
which establishes~\eqref{eq: performance error result 1}. Finally, note that:
\begin{align}
\expec\|\bcwt^{e'}_i\|^2=\expec\|\bcwt^{e'}_i-\bcwt_i^e+\bcwt_i^e\|^2&\leq\expec\|\bcwt^{e'}_i-\bcwt_i^e\|^2+\expec\|\bcwt_i^e\|^2+2|\expec(\bcwt^{e'}_i-\bcwt_i^e)^*\bcwt_i^e|\notag\\
&\leq\expec\|\bcwt^{e'}_i-\bcwt_i^e\|^2+\expec\|\bcwt_i^e\|^2+2\sqrt{\expec\|\bcwt^{e'}_i-\bcwt_i^e\|^2\expec\|\bcwt_i^e\|^2}
\end{align}
where we used $|\expec\bx|\leq\expec|\bx|$ from Jensen's inequality and where we applied Holder's inequality:
\begin{equation}
\expec|\bx^*\by|\leq(\expec|\bx|^p)^{\frac{1}{p}}(\expec|\by|^q)^{\frac{1}{q}}, \quad\text{when }1/p+1/q=1.
\end{equation}
Hence, from~\eqref{eq: mean-square error convergence result} and~\eqref{eq: performance error result 1} we get:
\begin{equation}
\label{eq: performance error result 3}
\limsup_{i\rightarrow \infty}(\expec\|\bcwt^{e'}_i\|^2-\expec\|\bcwt_i^e\|^2)\leq O(\mu^2)+O(\mu^{3/2})=O(\mu^{3/2})
\end{equation}
since $\mu^2<\mu^{3/2}$ for small $\mu\ll 1$, which establishes~\eqref{eq: performance error result 2}.

From~\eqref{eq: mean-square error convergence result} and~\eqref{eq: performance error result 3}, it follows that:
\begin{equation}
\limsup_{i\rightarrow \infty}\expec\|\bcwt^{e'}_i\|^2= O(\mu),
\end{equation}
and, therefore, the long-term approximate model~\eqref{eq: long term error model 2} is also mean-square stable.

\section{Low-rank approximation}
\label{app: low-rank approximation}
From~\eqref{eq: cBb definition}, we obtain:
\begin{equation}
\label{eq: definition of Btop}
\cB^\top=\left((\cV_{\epsilon}^e)^\top\right)^{-1}\left[\begin{array}{cc}
I_{hP}-\cD_{11}^\top&-\cD_{21}^\top\\
-\cD_{12}^\top&(\cJ^e_{\epsilon})^\top-\cD_{22}^\top
\end{array}\right](\cV_{\epsilon}^e)^\top
\end{equation}
\begin{equation}
\label{eq: definition of Bstar}
\cB^*=\left((\cV_{\epsilon}^e)^*\right)^{-1}\left[\begin{array}{cc}
I_{hP}-\cD_{11}^*&-\cD_{21}^*\\
-\cD_{12}^*&(\cJ^e_{\epsilon})^*-\cD_{22}^*
\end{array}\right](\cV_{\epsilon}^e)^*
\end{equation}
where the block matrices $\{\cD^\top_{mn},\cD^*_{mn}\}$ are all on the order of $\mu$ with:
\begin{align}
\cD^\top_{11}&=\mu\,(\cU^e)^\top(\cH^o)^\top[(\cU^e)^*]^\top=O(\mu),\label{eq: definition of D11 transpose}\\
\cD^*_{11}&=\cD_{11}=\mu\,(\cU^e)^*\cH^o\cU^e=O(\mu),\label{eq: definition of D11 transpose}
\end{align}
of dimensions $hP\times hP$. Substituting~\eqref{eq: definition of Btop} and~\eqref{eq: definition of Bstar} into~\eqref{eq: definition of the matrix cF} and using property $(A\otimes_b B)(C\otimes_bD)=AC\otimes_b BD,$
for block Kronecker products, we obtain:
\begin{equation}
\cF=\left((\cV^e_{\epsilon})^\top\otimes_b(\cV^e_{\epsilon})^*\right)^{-1}\cX\left((\cV^e_{\epsilon})^\top\otimes_b(\cV^e_{\epsilon})^*\right),
\end{equation}
 where we introduced:
\begin{equation}
\begin{split}
\cX\triangleq&\left[\begin{array}{cc}
I_{hP}-\cD_{11}^\top&-\cD_{21}^\top\\
-\cD_{12}^\top&(\cJ^e_{\epsilon})^\top-\cD_{22}^\top
\end{array}\right]\otimes_b\left[\begin{array}{cc}
I_{hP}-\cD_{11}^*&-\cD_{21}^*\\
-\cD_{12}^*&(\cJ^e_{\epsilon})^*-\cD_{22}^*
\end{array}\right].
\end{split}
\end{equation}
We partition $\cX$ into the following block structure:
\begin{equation}
\cX=\left[\begin{array}{cc}
\cX_{11}&\cX_{12}\\
\cX_{21}&\cX_{22}
\end{array}\right]
\end{equation}
where, for example, $\cX_{11}$ is $(hP)^2\times (hP)^2$ and is given by:
\begin{equation}
\cX_{11}\triangleq\left(I_{hP}-\cD_{11}^\top\right)\otimes\left(I_{hP}-\cD_{11}^*\right).
\end{equation}
Since
\begin{equation}
\label{eq: first equation for I-cF inverse}
(I-\cF)^{-1}\hspace{-1mm}=\left((\cV^e_{\epsilon})^\top\otimes_b(\cV^e_{\epsilon})^*\right)^{-1}\hspace{-1mm}\left(I-\cX\right)^{-1}\hspace{-1mm}\left((\cV^e_{\epsilon})^\top\otimes_b(\cV^e_{\epsilon})^*\right)
\end{equation}
we proceed to evaluate $I-\cF$. It follows that:
\begin{equation}
I-\cX=\left[\begin{array}{cc}
I_{(hP)^2}-\cX_{11}&-\cX_{12}\\
-\cX_{21}&I-\cX_{22}
\end{array}\right]
\end{equation}
and, in a manner similar to the way we assessed the size of the block matrices $\{\bcD_{mn,i-1}\}$ in the proof of Theorem~1 in Part~I~\cite{nassif2019adaptation}, we can verify that:
\begin{equation}
\begin{split}
I-\cX_{11}&=I_{(hP)^2}-\left(I_{hP}-\cD_{11}^\top\right)\otimes\left(I_{hP}-\cD_{11}^*\right)=O(\mu),\\
\cX_{12}&=O(\mu),\quad \cX_{21}=O(\mu),\quad I-\cX_{22}=O(1).
\end{split}
\end{equation}
\cred{The matrix $I-\cX$ is invertible since $I-\cF$ is invertible; this is because $\rho(\cF)=[\rho(\cB)]^2<1$.} Therefore, applying the block matrix inversion formula to $I-\cX$ we get:
\begin{equation}
\begin{split}
\left(I-\cX\right)^{-1}=\left[\begin{array}{cc}
(I_{(hP)^2}-\cX_{11})^{-1}&0\\
0&0
\end{array}\right]+\left[\begin{array}{cc}
\cY_{11}&\cY_{12}\\
\cY_{21}&\Delta^{-1}
\end{array}\right]
\end{split}
\end{equation}
where $\cY_{11}=(I-\cX_{11})^{-1}\cX_{12}\Delta^{-1}\cX_{21}(I-\cX_{11})^{-1}$, $\cY_{12}=(I-\cX_{11})^{-1}\cX_{12}\Delta^{-1}$, $\cY_{21}=\Delta^{-1}\cX_{21}(I-\cX_{11})^{-1}$, and $\Delta=(I-\cX_{22})-\cX_{21}(I-\cX_{11})^{-1}\cX_{12}$. The entries of $(I_{(hP)^2}-\cX_{11})^{-1}$ are $O(1/\mu)$, while the entries in the second matrix on the right-hand side of the above equation are $O(1)$ when the step-size is small. That is, we can write:
\begin{equation}
\left(I-\cX\right)^{-1}=\left[\begin{array}{c|c}
O(1/\mu)&O(1)\\
\hline
O(1)&O(1)
\end{array}\right].
\end{equation}
Moreover, since $O(1/\mu)$ dominates $O(1)$ for sufficiently small $\mu$, we can also write:
\begin{align}
\label{eq: inverse of I-cX for cF}
\left(I-\cX\right)^{-1}&=\left[\begin{array}{cc}
(I_{(hP)^2}-\cX_{11})^{-1}&0\\
0&0
\end{array}\right]+O(1)\notag\\
&=\left[\begin{array}{cc}
\left((I_{(hP)}\otimes\cD_{11}^*)+(\cD_{11}^\top\otimes I_{(hP)})\right)^{-1}&0\\
0&0
\end{array}\right]+O(1)\notag\\
&=\left[\begin{array}{c}
I_{(hP)^2}\\
0
\end{array}\right]Z^{-1}\left[\begin{array}{cc}
I_{(hP)^2}&0
\end{array}\right]+O(1).
\end{align}
Substituting~\eqref{eq: inverse of I-cX for cF} into~\eqref{eq: first equation for I-cF inverse} we arrive at~\eqref{eq: low rank approximation of I-F inverse}. Since $Z=O(\mu)$, we conclude that~\eqref{eq: order of inverse of I-cF} holds. We also conclude that~\eqref{eq: order of inverse of I-cFb} holds since:
\begin{equation}
(I-\cFb)^{-1}
=(I-\cX)^{-1}.
\end{equation}

\section{Proof of Theorem~\ref{theo: Network limiting performance}}
\label{app: proof of theorem 3}
Consider the long-term model~\eqref{eq: long term error model 2}. Conditioning both sides of~\eqref{eq: long term error model 2} on $\bcF_{i-1}$, invoking the conditions on the gradient noise process from Assumption~\ref{assump: gradient noise}, and computing the conditional expectation, we obtain:
\begin{equation}
\label{eq: mean long term error model}
\expec[\bcwt_i^{e'}|\bcF_{i-1}]=\cB\bcwt_{i-1}^{e'}+\mu\cA^e b^e,
\end{equation}
where the term involving $\bs_i^e(\bcw^e_{i-1})$ is eliminated because $\expec[\bs_i^e|\bcF_{i-1}]=0$. Taking expectations again we arrive at:
\begin{equation}
\label{eq: mean error recursion of the long-term error model}
\expec\bcwt_i^{e'}=\cB\left(\expec\bcwt_{i-1}^{e'}\right)+\mu\cA^e b^e.
\end{equation}
 Since recursion~\eqref{eq: long term error model 2} includes a constant driving term $\mu\cA^e b^e$, we introduce the centered variable  
$\by_i\triangleq\bcwt^{e'}_i-\expec\bcwt^{e'}_i.$
Subtracting~\eqref{eq: mean error recursion of the long-term error model} from~\eqref{eq: long term error model 2}, we find that $\by_i$ satisfies the following recursion:
\begin{equation}
\label{eq: recursion for bz i long term}
\by_i=\cB\by_{i-1}-\mu\cA^e\bs_{i}^e(\bcw_{i-1}).
\end{equation}
Although we are interested in evaluating $\limsup_{i\rightarrow\infty}\expec\|\bcwt^{e'}_i\|^2$, we can still rely on $\by_i$ since it holds for $i\gg 1$:
\begin{equation}
\label{eq: relation between zi and w'i}
\expec\|\bz_i\|^2=\expec\|\bcwt^{e'}_i\|^2-\|\expec\bcwt^{e'}_i\|^2=\expec\|\bcwt^{e'}_i\|^2+O(\mu^2),
\end{equation}
where we used the fact that $\limsup_{i\rightarrow\infty}\|\expec\bcwt^{e'}_i\|=O(\mu)$ (see~Appendix~\ref{sec: Stability of first-order error moment}). Therefore, from~\eqref{eq: performance error result 2} and~\eqref{eq: relation between zi and w'i}, we obtain:
\begin{equation}
\label{eq: steady state of wti in terms of the steady state of zi}
\limsup_{i\rightarrow\infty}\frac{1}{{h}N}\expec\|\bcwt_i^e\|^2=\limsup_{i\rightarrow\infty}\frac{1}{{h}N}\expec\|\by_i\|^2+O(\mu^{3/2}).
\end{equation}

Let $\Sigma$ denote an arbitrary Hermitian positive semi-definite matrix that we are free to choose. Equating the squared weighted values of both sides of~\eqref{eq: recursion for bz i long term} and taking expectations conditioned on the past history gives:
\begin{equation}
\expec[\|\by_{i}\|^2_{\Sigma}|\bcF_{i-1}]=\|\by_{i-1}\|^2_{\cB^*\Sigma\cB}+\mu^2\expec[\|\bs^e_{i}\|^2_{(\cA^e)^*\Sigma\cA^e}|\bcF_{i-1}].
\end{equation}
Taking expectations again, we get:
\begin{equation}
\expec\|\by_{i}\|^2_{\Sigma}=\expec\|\by_{i-1}\|^2_{\cB^*\Sigma\cB}+\mu^2\expec\|\bs^e_{i}\|^2_{(\cA^e)^*\Sigma\cA^e}.\label{eq: weighted mean-square error of zi}
\end{equation}
From~\eqref{eq: zero bound on the difference of covariances} and using same arguments as in~\cite[pp.~586]{sayed2014adaptation}, we can rewrite the second term on the R.H.S.  of~\eqref{eq: weighted mean-square error of zi} as:
\begin{equation}
\begin{split}
\mu^2 \expec\|\bs^e_{i}\|^2_{(\cA^e)^*\Sigma\cA^e}&=\mu^2\tr\left(\Sigma\cA^e\expec[\bs^e_{i}\bs_{i}^{e*}](\cA^e)^*\right)\\
&=\tr(\Sigma\cY)+\tr(\Sigma)\cdot O(\mu^{2+\gamma_m}). 
\end{split}
\end{equation}
for $i\gg1$. Therefore, we obtain:
\begin{equation}
\limsup_{i\rightarrow\infty}\expec\|\by_{i}\|^2_{\Sigma-\cB^*\Sigma\cB}=\tr(\Sigma\cY)+\tr(\Sigma)\cdot O(\mu^{2+\gamma_m}).\label{eq: weighted mean-square error of zi 8}
\end{equation}

In order to reduce the weighting matrix on the mean-square value of $\bz_i$ in~\eqref{eq: weighted mean-square error of zi 8} to the identity, we need to select $\Sigma$ as the solution to the following Lyapunov equation:
\begin{equation}
\label{eq: Lyapunov equation}
\Sigma-\cB^*\Sigma\cB=I_{hM}.
\end{equation}
This equation has a unique Hermitian non-negative definite solution $\Sigma$~\cite[pp.~772]{sayed2014adaptation} since the matrix $\cB$ is stable for sufficiently small step-size. Now, by applying the block vectorization operation to both sides of~\eqref{eq: Lyapunov equation} and by using the property that:
\begin{equation}
\bvc(ACB)=(B^\top\otimes_b A)\bvc(C),\label{eq: block kronecker property 7}
\end{equation}
we find that:
\begin{equation}
\label{eq: Lyapunov equation 2}
\bvc(\Sigma)=(I-\cF)^{-1}\bvc(I_{hM})
\end{equation}
in terms of the matrix $\cF$ defined in~\eqref{eq: definition of the matrix cF}.

Now, substituting $\Sigma$ in~\eqref{eq: Lyapunov equation 2} into~\eqref{eq: weighted mean-square error of zi 8}, we obtain $\expec\|\by_i\|^2$ on the left-hand side while the term $\tr(\Sigma\cY)$ on the right-hand side becomes:
\begin{align}
\label{eq: first equetion}
\tr(\Sigma\cY)=(\bvc(\cY^\top))^\top(I-\cF)^{-1}\bvc(I_{hM}).
\end{align}
where we used the property that:
\begin{equation}
\tr( A B)=[\bvc( B^\top)]^\top\bvc( A),\label{eq: block kronecker property 6}
\end{equation}
 Using the fact that $(I-\cF)^{-1}=O(1/\mu)$ and following similar arguments as in~\cite[pp.~590]{sayed2014adaptation}, we can show that:
\begin{equation}
\label{eq: second equetion}
\tr(\Sigma)\cdot O(\mu^{2+\gamma_m})=O(\mu^{1+\gamma_m}).
\end{equation}
Replacing~\eqref{eq: first equetion} and~\eqref{eq: second equetion} into~\eqref{eq: weighted mean-square error of zi 8} gives~\eqref{eq: network error variance}. Observe that the first term on the R.H.S. of~\eqref{eq: network error variance} is $O(\mu)$ since $\|\cY\|=O(\mu^2)$ and $\|(I-\cF)^{-1}\|=O(1/\mu)$. Therefore, this term dominates the factor $O(\mu^{1+\gamma_m})$. 

Since $\cF$ is a stable matrix for sufficiently small step-sizes, we can employ the expansion
$(I-\cF)^{-1}=I+\cF+\cF^2+\cF^3+\ldots,$ 
replace $\cF$ by~\eqref{eq: definition of the matrix cF}, and use properties~\eqref{eq: block kronecker property 6} and~\eqref{eq: block kronecker property 7} to write the first term on the R.H.S. of~\eqref{eq: network error variance} as $\sum_{n=0}^{\infty}\tr(\cB^n\cY(\cB^*)^n)$.

Now, in order to establish~\eqref{eq: final result for MSD}, we shall use the low-rank approximation~\eqref{eq: low rank approximation of I-F inverse}. Using definition~\eqref{eq: MSD definition} and~\eqref{eq: network error variance}, we obtain:
\begin{equation}
\text{MSD}=\frac{\mu}{hN}\lim_{\mu\rightarrow 0}\limsup_{i\rightarrow\infty}\frac{1}{\mu} 
({\bvc}(\cY^\top))^\top(I-\cF)^{-1}{\bvc}(I_{hM})
\label{eq: definition of the network MSD final}
\end{equation}
From~\eqref{eq: low rank approximation of I-F inverse} we get:
\begin{equation}
\label{eq: relation intermediaire finale 0}
\begin{split}
&({\bvc}(\cY^\top))^\top(I-\cF)^{-1}{\bvc}(I_{hM})=({\bvc}(\cY^\top))^\top\left([(\cU^e)^*]^\top\otimes_b\cU^e\right)Z^{-1}\left((\cU^e)^\top\otimes_b(\cU^e)^*\right){\bvc}(I_{hM})+O(\mu^2).
\end{split}
\end{equation}
Using property~\eqref{eq: block kronecker property 7}, it is straightforward to verify that the last three terms combine into the following result:
\begin{equation}
\left((\cU^e)^\top\otimes_b(\cU^e)^*\right){\bvc}(I_{hM})=\bvc(I_{hP})=\vc(I_{hP}),
\end{equation}
Let us therefore evaluate the matrix vector product 
$x\triangleq Z^{-1}\vc(I_{hP}).$
 Using the definition~\eqref{eq: definition of the matrix Z} for $Z$, the vector $x$ is therefore the unique solution to the linear system of equations:
\begin{equation}
\label{eq: relation intermediaire finale}
(I_{hP}\otimes\cD_{11}^*)x+(\cD_{11}^\top\otimes I_{hP})x=\vc(I_{hP}).
\end{equation}
Let $X=\text{unvec}(x)$ denote the $hP\times hP$ matrix whose vector representation is $x$. Applying to each of the terms appearing on the left-hand side of the above expression the Kronecker product property~\eqref{eq: block kronecker property 7}, albeit using $\vc$ instead of $\bvc$ operation, we find that
$(I_{hP}\otimes\cD_{11}^*)x=\vc(\cD_{11}^* X),$ and $
(\cD_{11}^\top\otimes I_{hP})x=\vc( X\cD_{11}).$ 
We conclude from these equalities and from~\eqref{eq: relation intermediaire finale} that $X$ is the unique solution to the (continuous-time) Lyapunov equation $
\cD_{11}^* X+X\cD_{11}=I_{hP}.$ 
Since $\cD_{11}$ in~\eqref{eq: definition cD 11} is Hermitian, we obtain:
\begin{equation}
X=\frac{1}{2}\cD_{11}^{-1}=\frac{1}{2\mu}\left((\cU^e)^*\cH^o\cU^e\right)^{-1}.\label{eq: computation of X}
\end{equation}
Therefore, substituting into~\eqref{eq: relation intermediaire finale 0} gives:
\begin{align}
({\bvc}(\cY^\top))^\top(I-\cF)^{-1}{\bvc}(I_{hM})&=({\bvc}(\cY^\top))^\top\cred{\left([(\cU^e)^*]^\top\otimes_b\cU^e\right)\vc(X)}+O(\mu^2)\notag\\
&\hspace{-1mm}\overset{\eqref{eq: block kronecker property 7}}=({\bvc}(\cY^\top))^\top\bvc(\cU^e X(\cU^e)^*)+O(\mu^2)\notag\\
&\hspace{-1mm}\overset{\eqref{eq: block kronecker property 6}}=\tr(\cU^e X(\cU^e)^*\cY)+O(\mu^2)=\tr((\cU^e)^*\cY\cU^e X)+O(\mu^2)\notag\\
&\hspace{-1mm}\overset{\eqref{eq: definition of the matrix cY}}=\mu^2\tr((\cU^e)^*\cA^e\cS(\cA^e)^*\cU^e X)+O(\mu^2)\notag\\
&=\mu^2\tr((\cU^e)^*\cS\cU^e X)+O(\mu^2)\notag\\
&\hspace{-1mm}\overset{\eqref{eq: computation of X}}=\frac{\mu}{2}\tr\left(\left((\cU^e)^*\cH^o\cU^e\right)^{-1}\left((\cU^e)^*\cS\cU^e\right)\right)+O(\mu^2).
\end{align}
where we used the fact that $(\cU^e)^*\cA^e=(\cU^e)^*$. Now substituting the above expression into the right-hand side of~\eqref{eq: definition of the network MSD final} and computing the limit as $\mu\rightarrow 0$, we arrive at expression~\eqref{eq: final result for MSD}.

\section{Projection onto $\Omega_1$ in~\eqref{eq: set Omega 1} and $\Omega_2$ in~\eqref{eq: set Omega 2}}
\label{app: projection onto Omega 1}
The closed convex set $\Omega_1$ in~\eqref{eq: set Omega 1} can be rewritten alternatively as:
\begin{equation}
\label{eq: set Omega 1 alternative}
\Omega_1=\{\cA|\cA\,\cU=\cU,\cU^\top\cA=\cU^\top,\cA=\cA^\top\},
\end{equation}
and the projection onto it is given by:
\begin{equation}
\label{eq: projection Omega 1 alternative}
\Pi_{\Omega_1}(\cD)=
\left\lbrace
\begin{array}{cl}
\arg\min\limits_\cA&\frac{1}{2}\|\cA-\cD\|^2_{\text{F}}\\
\st&\cA\,\cU=\cU,\cU^\top\cA=\cU^\top,\cA=\cA^\top.
\end{array}
\right.
\end{equation}
The Lagrangian of the convex optimization problem in~\eqref{eq: projection Omega 1 alternative} is defined as:
\begin{equation}
\begin{split}
\cL(\cA;X,Y,W)=\frac{1}{2}\|\cA-\cD\|^2_{\text{F}}+\tr(X^\top(\cA\cU-\cU))+\\
~~\tr(Y^\top(\cU^\top\cA-\cU^\top))+\tr(Z^\top(\cA-\cA^\top)),
\end{split}
\end{equation}
where $X\in\mathbb{R}^{M\times P}$, $Y\in\mathbb{R}^{P\times M}$, and $Z\in\mathbb{R}^{M\times M}$ are the matrices of Lagrange multipliers. From the Karush-Kuhn-Tucker (KKT) conditions, we obtain at the optimum $(\cA^o;X^o,Y^o,Z^o)$:
\begin{align}
\cA^o\cU&=\cU,\label{eq: projection first condition}\\
\cU^\top\cA^o&=\cU^\top,\label{eq: projection second condition}\\
\cA^o&=(\cA^o)^\top,\label{eq: projection third condition}\\
\nabla_{\cA}\cL&=\cA^o-\cD+X^o\cU^\top+\cU Y^o+Z^o-(Z^o)^\top=0. \label{eq: projection fourth condition}
\end{align}
From~\eqref{eq: projection fourth condition}, we obtain:
\begin{equation}
\label{eq: first expression for A o}
\cA^o=\cD-X^o\cU^\top-\cU Y^o-Z^o+(Z^o)^\top.
\end{equation}
Multiplying both sides of~\eqref{eq: first expression for A o} by $\cU$ and using the fact that $\cU^\top\cU=I$ from Assumption~\ref{assump: matrix cU}, we obtain:
\begin{equation}
\label{eq: second expression for A o}
\cA^o\cU=\cD\cU-X^o-\cU Y^o\cU-Z^o\cU+(Z^o)^\top\cU.
\end{equation}
Combining the previous expression with~\eqref{eq: projection first condition}, we get:
\begin{equation}
\label{eq: third expression for A o}
X^o=\cD\cU-\cU Y^o\cU-Z^o\cU+(Z^o)^\top\cU-\cU.
\end{equation}
Replacing~\eqref{eq: third expression for A o} into~\eqref{eq: first expression for A o} and using the fact that $\cP_{\cu}=\cU\cU^\top$, we arrive at:
\begin{equation}
\label{eq: fourth expression for A o}
\begin{split}
\cA^o=&\cD-\cD\cP_{\cu}+\cU Y^o\cP_{\cu}+Z^o\cP_{\cu}-(Z^o)^\top\cP_{\cu}+\cP_{\cu}-\cU Y^o-Z^o+(Z^o)^\top.
\end{split}
\end{equation}
Pre-multiplying both sides of the previous equation by $\cU^\top$ and using the fact that $\cU^\top\cU=I$, we obtain:
\begin{equation}
\label{eq: fifth expression for A o}
\begin{split}
\cU^\top\cA^o=\cU^\top\cD-\cU^\top\cD\cP_{\cu}+Y^o\cP_{\cu}+\cU^\top Z^o\cP_{\cu}-\cU^\top(Z^o)^\top\cP_{\cu}+\cU^\top- Y^o-\cU^\top Z^o+\cU^\top (Z^o)^\top.
\end{split}
\end{equation}
Combining the previous expression with~\eqref{eq: projection second condition}, we arrive at:
\begin{equation}
\label{eq: sixth expression for A o}
\begin{split}
\cU^\top\cD-\cU^\top\cD\cP_{\cu}+Y^o\cP_{\cu}+\cU^\top Z^o\cP_{\cu}-\cU^\top(Z^o)^\top\cP_{\cu}- Y^o-\cU^\top Z^o+\cU^\top (Z^o)^\top=0.
\end{split}
\end{equation}
Pre-multiplying both sides of the previous equation by $\cU$ and using the fact that $\cP_{\cu}=\cU\cU^\top$, we obtain:
\begin{equation}
\begin{split}
\label{eq: seventh expression for A o}
\cU Y^o\cP_{\cu}- \cU Y^o=&-\cP_{\cu}\cD+\cP_{\cu}\cD\cP_{\cu}-\cP_{\cu} Z^o\cP_{\cu}+\cP_{\cu}(Z^o)^\top\cP_{\cu}+\cP_{\cu} Z^o-\cP_{\cu} (Z^o)^\top.
\end{split}
\end{equation}
Replacing~\eqref{eq: seventh expression for A o} into~\eqref{eq: fourth expression for A o}, we arrive at:
\begin{equation}
\label{eq: eight expression for A o}
\begin{split}
\cA^o=&
(I-\cP_{\cu})\cD(I-\cP_{\cu})-(I-\cP_{\cu})(Z^o-(Z^o)^\top)(I-\cP_{\cu})+\cP_{\cu},
\end{split}
\end{equation}
and thus,
\begin{equation}
\label{eq: nine expression for A o}
\cA^{o\top}=(I-\cP_{\cu})\cD^\top(I-\cP_{\cu})+(I-\cP_{\cu})(Z^o-(Z^o)^\top)(I-\cP_{\cu})+\cP_{\cu}
\end{equation}
Combining~\eqref{eq: projection third condition} and the previous two equations, we obtain:
\begin{equation}
\label{eq: ten expression for A o}
(I-\cP_{\cu})\left(\frac{\cD-\cD^\top}{2}\right)(I-\cP_{\cu})\hspace{-1mm}=\hspace{-1mm}(I-\cP_{\cu})(Z^o-(Z^o)^\top)(I-\cP_{\cu})
\end{equation}
Replacing the previous equation into~\eqref{eq: eight expression for A o}, we arrive at:
\begin{equation}
\label{eq: final version of the projection}
\begin{split}
\Pi_{\Omega_1}(\cD)&=(I-\cP_{\cu})\left(\cD-\frac{\cD-\cD^\top}{2}\right)(I-\cP_{\cu})+\cP_{\cu}\\
&=(I-\cP_{\cu})\left(\frac{\cD+\cD^\top}{2}\right)(I-\cP_{\cu})+\cP_{\cu}.
\end{split}
\end{equation}

Now, projecting a symmetric matrix $\cC$ onto $\Omega_2$ in~\eqref{eq: set Omega 2} is given by:
\begin{equation}
\begin{split}
\Pi_{\Omega_2}(\cC)&=
\left\lbrace
\begin{array}{cl}
\arg\min\limits_\cA&\frac{1}{2}\|\cA-\cC\|^2_{\text{F}}\\
\st&\|\cA-\cP_{\cu}\|\leq 1-\epsilon
\end{array}
\right.\\
&=\cP_{\cu}+
\left\lbrace
\begin{array}{cl}
\arg\min\limits_\cY&\frac{1}{2}\|\cY-(\cC-\cP_{\cu})\|^2_{\text{F}}\\
\st&\|\cY\|\leq 1-\epsilon
\end{array}
\right.\\
&=\cP_{\cu}+\Pi_{\Omega_3}(\cC-\cP_{\cu})
\end{split}
\end{equation}
where $\Omega_3\triangleq\{\cA|\|\cA\|\leq1-\epsilon\}$. In order to project the symmetric matrix $\cC-\cP_{\cu}$ onto $\Omega_3$, we need to compute its eigenvalue decomposition $\cC-\cP_{\cu}=\sum_{m=1}^M\lambda_mv_mv_m^\top$ and then threshold the eigenvalues $\{\lambda_m\}_{m=1}^M$ to have absolute magnitude at most $1-\epsilon$~\cite[pp.~191--194]{parikh2014proximal}. Thus we obtain:
\begin{equation}
\Pi_{\Omega_2}(\cC)=\cP_{\cu}+\sum_{m=1}^M\beta_mv_mv_m^\top,
\end{equation}
where:
\begin{equation}
\beta_m=\left\lbrace
\begin{array}{ll}
-1+\epsilon, &\text{if } \lambda_m<-1+\epsilon,\\
\lambda_m, &\text{if }  |\lambda_m|<1-\epsilon,\\
1-\epsilon, &\text{if }  \lambda_m>1-\epsilon.
\end{array}
\right.
\end{equation}
Now, replacing the matrix $\cC$ by~\eqref{eq: final version of the projection}, we obtain~\eqref{eq: projection onto omega 2}.

\section{Proof of Lemma~\ref{lemma: one step projection}}
\label{app: proof of lemma 4}
In order to establish Lemma~\ref{lemma: one step projection}, we first need to introduce Lemmas~\ref{lemma: vi convex} and~\ref{lemma: vi affine}.
\begin{lemma}{}
\label{lemma: vi convex}
Let $\Omega$ denote a closed convex set. 
 For any $\cC\notin\Omega$, $\cA^o=\Pi_{\Omega}(\cC)$ if and only if $\langle\cC-\cA^o,\cA-\cA^o\rangle\leq 0$, $\forall \cA\in\Omega$ where $\langle X,Y\rangle=\emph{\tr}(X^\top Y)$. 
\end{lemma}
\begin{proof}($\Rightarrow$) Let $\cA^o=\Pi_{\Omega}(\cC)$ for any given $\cC\notin\Omega$, that is, suppose that $\cA^o$ is the unique solution to the optimization problem. Let $\cA\in\Omega$ be such that $\cA\neq\cA^o$. Let $\alpha\in(0,1)$. Since $\Omega$ is convex, $(1-\alpha)\cA^o+\alpha\cA=\cA^o+\alpha(\cA-\cA^o)\in\Omega$. By the assumed optimality of $\cA^o$, we must have:
\begin{equation}
\begin{split}
\|\cC-\cA^o\|^2_{\text{F}}&\leq\|\cC-[\cA^o+\alpha(\cA-\cA^o)]\|^2_\text{F}\\
&=\|\cC-\cA^o\|^2_{\text{F}}+\alpha^2\|\cA-\cA^o\|^2_\text{F}-2\alpha\langle\cC-\cA^o,\cA-\cA^o\rangle,
\end{split}
\end{equation}
and we obtain:
\begin{equation}
\label{eq: proof lemma eq 1}
\langle\cC-\cA^o,\cA-\cA^o\rangle\leq\frac{\alpha}{2}\|\cA-\cA^o\|^2_\text{F}.
\end{equation}
Now, note that~\eqref{eq: proof lemma eq 1} holds for any $\alpha\in(0,1)$. Since the RHS of~\eqref{eq: proof lemma eq 1} can be made arbitrarily small for a given $\cA$, the LHS can not be strictly positive. Thus, we conclude as desired:
\begin{equation}
\label{eq: proof lemma eq 2}
\langle\cC-\cA^o,\cA-\cA^o\rangle\leq 0,~~\forall \cA\in\Omega.
\end{equation}
($\Leftarrow$) Let $\cA^o\in\Omega$ be such that $\langle\cC-\cA^o,\cA-\cA^o\rangle\leq 0,\forall \cA\in\Omega$. We shall show that it must be the optimal solution. Let $\cA\in\Omega$ and $\cA\neq\cA^o$. We have:
\begin{equation}
\begin{split}
\|\cC-\cA\|^2_{\text{F}}-\|\cC-\cA^o\|^2_{\text{F}}&=\|\cC-\cA^o+\cA^o-\cA\|^2_{\text{F}}-\|\cC-\cA^o\|^2_{\text{F}}\\
&=\|\cC-\cA^o\|^2_{\text{F}}+\|\cA^o-\cA\|^2_{\text{F}}-2\langle\cC-\cA^o,\cA-\cA^o\rangle-\|\cC-\cA^o\|^2_{\text{F}}>0.
\end{split}
\end{equation}
Hence, $\cA^o$ is the optimal solution to the optimization problem, and thus $\cA^o=\Pi_{\Omega}(\cC)$ by definition.
\end{proof}

\begin{lemma}{}
\label{lemma: vi affine}
If $\Omega$ is further affine,
 then, for any $\cC\notin\Omega$, $\cA^o=\Pi_{\Omega}(\cC)$ if and only if $\langle\cC-\cA^o,\cA-\cA^o\rangle= 0$, $\forall \cA\in\Omega$.
\end{lemma}
\begin{proof}($\Rightarrow$) Let $\cA^o=\Pi_{\Omega}(\cC)$ for any given $\cC\notin\Omega$, that is, suppose that $\cA^o$ is the unique solution to the optimization problem. Let $\cA\in\Omega$ be such that $\cA\neq\cA^o$. Let $\alpha\in\mathbb{R}$. Since $\Omega$ is affine, $(1-\alpha)\cA^o+\alpha\cA=\cA^o+\alpha(\cA-\cA^o)\in\Omega$. By the assumed optimality of $\cA^o$, we must have:
\begin{equation}
\begin{split}
\|\cC-\cA^o\|^2_{\text{F}}&\leq\|\cC-[\cA^o+\alpha(\cA-\cA^o)]\|^2_\text{F}\\
&=\|\cC-\cA^o\|^2_{\text{F}}+\alpha^2\|\cA-\cA^o\|^2_\text{F}-2\alpha\langle\cC-\cA^o,\cA-\cA^o\rangle,
\end{split}
\end{equation}
and we obtain:
\begin{equation}
\label{eq: proof lemma eq 3}
2\alpha\langle\cC-\cA^o,\cA-\cA^o\rangle\leq\alpha^2\|\cA-\cA^o\|^2_\text{F}.
\end{equation}
If $\alpha\geq 0$, we obtain:
\begin{equation}
\label{eq: proof lemma eq 4}
\langle\cC-\cA^o,\cA-\cA^o\rangle\leq\frac{\alpha}{2}\|\cA-\cA^o\|^2_\text{F}.
\end{equation}
Now, note that~\eqref{eq: proof lemma eq 4} holds for any $\alpha\geq 0$. Since the RHS of~\eqref{eq: proof lemma eq 4} can be made arbitrarily small for a given $\cA$, the LHS can not be strictly positive. Thus, we conclude:
\begin{equation}
\label{eq: proof lemma eq 5}
\langle\cC-\cA^o,\cA-\cA^o\rangle\leq 0,~~\forall \cA\in\Omega.
\end{equation}
If $\alpha\leq 0$, we obtain:
\begin{equation}
\label{eq: proof lemma eq 6}
\langle\cC-\cA^o,\cA-\cA^o\rangle\geq\frac{\alpha}{2}\|\cA-\cA^o\|^2_\text{F}.
\end{equation}
Now, note that~\eqref{eq: proof lemma eq 6} holds for any $\alpha\leq 0$. Since the RHS of~\eqref{eq: proof lemma eq 6} can be made arbitrarily large for a given $\cA$, the LHS can not be strictly negative. Thus, we conclude:
\begin{equation}
\label{eq: proof lemma eq 7}
\langle\cC-\cA^o,\cA-\cA^o\rangle\geq 0,~~\forall \cA\in\Omega.
\end{equation}
Combining~\eqref{eq: proof lemma eq 5} and~\eqref{eq: proof lemma eq 7}, we conclude as desired:
\begin{equation}
\label{eq: proof lemma eq 8}
\langle\cC-\cA^o,\cA-\cA^o\rangle= 0,~~\forall \cA\in\Omega.
\end{equation}
($\Leftarrow$) Let $\cA^o\in\Omega$ be such that $\langle\cC-\cA^o,\cA-\cA^o\rangle= 0,\forall \cA\in\Omega$. We shall show that it must be the optimal solution. Let $\cA\in\Omega$ and $\cA\neq\cA^o$. We have:
\begin{equation}
\begin{split}
\|\cC-\cA\|^2_{\text{F}}-\|\cC-\cA^o\|^2_{\text{F}}&=\|\cC-\cA^o+\cA^o-\cA\|^2_{\text{F}}-\|\cC-\cA^o\|^2_{\text{F}}\\
&=\|\cC-\cA^o\|^2_{\text{F}}+\|\cA^o-\cA\|^2_{\text{F}}-2\langle\cC-\cA^o,\cA-\cA^o\rangle-\|\cC-\cA^o\|^2_{\text{F}}>0.
\end{split}
\end{equation}
Hence, $\cA^o$ is the optimal solution to the optimization problem, and thus $\cA^o=\Pi_{\Omega}(\cC)$ by definition.
\end{proof}
Now we prove Lemma~\ref{lemma: one step projection}. Let $Y=\Pi_{\Omega_1}(\cC)$. From Lemma~\ref{lemma: vi affine}, we have:
\begin{equation}
\label{eq: proof lemma eq 9}
\langle\cC-Y,\cA-Y\rangle= 0,\quad \forall \cA\in\Omega_1.
\end{equation}
Let $Z=\Pi_{\Omega_2}(Y)$. From Lemma~\ref{lemma: vi convex}, we have:
\begin{equation}
\label{eq: proof lemma eq 10}
\langle Y-Z,\cA-Z\rangle\leq 0,\quad \forall \cA\in\Omega_2.
\end{equation}
For $Z=\Pi_{\Omega_2}(\Pi_{\Omega_1}(\cC))$ to be the projection of $\cC$ onto $\Omega_1\cap\Omega_2$, we need to show from Lemma~\ref{lemma: vi convex} that:
\begin{equation}
\label{eq: proof lemma eq 11}
\langle \cC-Z,\cA-Z\rangle\leq 0,\quad \forall \cA\in\Omega_1\cap\Omega_2,
\end{equation}
under the conditions in Lemma~\ref{lemma: one step projection}. For any $\cA\in\Omega_1\cap\Omega_2$, we have:
\begin{equation}
\begin{split}
&\langle \cC-Z,\cA-Z\rangle=\langle \cC-Y+Y-Z,\cA-Z\rangle\\
&=\langle \cC-Y,\cA-Z\rangle+\langle Y-Z,\cA-Z\rangle\\
&=\langle \cC-Y,\cA-Y+Y-Z\rangle+\langle Y-Z,\cA-Z\rangle\\
&=\underbrace{\langle \cC-Y,\cA-Y\rangle}_{=0~\text{from~}\eqref{eq: proof lemma eq 9}}
-\hspace{-2.5mm}\underbrace{\langle \cC-Y,Z-Y\rangle}_{=0~\text{from~\eqref{eq: proof lemma eq 9} and~}Z\in\Omega_1}\hspace{-2.5mm}
+\underbrace{\langle Y-Z,\cA-Z\rangle}_{\leq 0~\text{from~}\eqref{eq: proof lemma eq 10}}\leq 0
\end{split}
\end{equation}
which concludes the proof.

\section{Stability of fourth-order error moment}
\label{sec: stability of fourth-order error moment}
In this Appendix, we show that, under the same settings of Theorem~1 in~\cite{nassif2019adaptation} with the second-order moment condition~\eqref{eq: gradient noise mean square condition} replaced by the fourth-order moment condition~\eqref{eq: gradient noise fourth moment condition}, the fourth-order moment of the network error vector is stable for sufficiently small $\mu$, namely, \eqref{eq: fourth order moment result} holds for small enough $\mu$. We start by recalling that for any two complex column vectors $x$ and $y$, it holds that 
$\|x+y\|^4\leq\|x\|^4+3\|y\|^4+8\|x\|^2\|y\|^2+4\|x\|^2\text{Re}(x^* y)$~\cite[pp.~523]{sayed2014adaptation}.
 Applying this inequality to~eq. (60) in~\cite{nassif2019adaptation}, conditioning on $\bcF_{i-1}$, computing the expectations of both sides, using Assumption~\ref{assump: gradient noise}, taking expectations again, and exploiting the convexity of $\|x\|^4$, we conclude that:
\begin{align}
\expec\|\bcwb_{i}^e\|^4&\leq\expec\|(I_{hP}-\bcD_{11,i-1})\bcwb_{i-1}^e-\bcD_{12,i-1}\bcwc^e_{i-1}\|^4+3\expec\|\bsb_{i}^e\|^4+\notag\\
&\qquad8\expec\left(\|(I_{hP}-\bcD_{11,i-1})\bcwb_{i-1}^e-\bcD_{12,i-1}\bcwc^e_{i-1}\|^2\right)\left(\expec\|\bsb_{i}^e\|^2\right)\notag\\
&=\expec\|(1-t)\frac{1}{1-t}(I_{hP}-\bcD_{11,i-1})\bcwb_{i-1}^e-t\frac{1}{t}\bcD_{12,i-1}\bcwc^e_{i-1}\|^4+3\expec\|\bsb_{i}^e\|^4+\notag\\
&\qquad8\expec\left(\|(1-t)\frac{1}{1-t}(I_{hP}-\bcD_{11,i-1})\bcwb_{i-1}^e-t\frac{1}{t}\bcD_{12,i-1}\bcwc^e_{i-1}\|^2\right)\left(\expec\|\bsb_{i}^e\|^2\right)\notag\\
&\leq\frac{1}{(1-t)^3}\expec\left[\|I_{hP}-\bcD_{11,i-1}\|^4\|\bcwb_{i-1}^e\|^4\right]+\frac{1}{t^3}\expec\left[\|\bcD_{12,i-1}\|^4\|\bcwc^e_{i-1}\|^4\right]+3\expec\|\bsb_{i}^e\|^4+\notag\\
&\qquad8\left(\expec\|\bsb_{i}^e\|^2\right)\left(\frac{1}{1-t}\expec\left[\|I_{hP}-\bcD_{11,i-1}\|^2\|\bcwb_{i-1}^e\|^2\right]+\frac{1}{t}\expec\left[\|\bcD_{12,i-1}\|^2\|\bcwc^e_{i-1}\|^2\right]\right)\notag\\
&{\leq}\frac{(1-\mu\sigma_{11})^4}{(1-t)^3}\expec\|\bcwb_{i-1}^e\|^4+\frac{\mu^4\sigma_{12}^4}{t^3}\expec\|\bcwc^e_{i-1}\|^4+3\expec\|\bsb_{i}^e\|^4+\notag\\
&\qquad\qquad8\left(\expec\|\bsb_{i}^e\|^2\right)\left(\frac{(1-\mu\sigma_{11})^2}{1-t}\expec\|\bcwb_{i-1}^e\|^2+\frac{\mu^2\sigma_{12}^2}{t}\expec\|\bcwc^e_{i-1}\|^2\right),
\end{align}
for any arbitrary positive number $t\in(0,1)$. In the last inequality we used the bounds (115) and (116) in~\cite{nassif2019adaptation}. By selecting $t=\mu\sigma_{11}$, we arrive at:
\begin{align}
\label{eq: fourth order bcwbc}
\expec\|\bcwb_{i}^e\|^4&\leq(1-\mu\sigma_{11})\expec\|\bcwb_{i-1}^e\|^4+\frac{\mu\sigma_{12}^4}{\sigma_{11}^3}\expec\|\bcwc^e_{i-1}\|^4+3\expec\|\bsb_{i}^e\|^4+\notag\\
&\qquad\qquad8\left(\expec\|\bsb_{i}^e\|^2\right)\left((1-\mu\sigma_{11})\expec\|\bcwb_{i-1}^e\|^2+\frac{\mu\sigma_{12}^2}{\sigma_{11}}\expec\|\bcwc^e_{i-1}\|^2\right).
\end{align}
Applying similar arguments for relation~(61) in~\cite{nassif2019adaptation} and using the relation 
$\|a+b+c\|^4\leq27\|a\|^4+27\|b\|^4+27\|c\|^4$,
we obtain:
\begin{align}
\expec\|\bcwc^e_{i}\|^4&\leq\expec\|(\cJ^e_{\epsilon}-\bcD_{22,i-1})\bcwc^e_{i-1}-\bcD_{21,i-1}\bcwb_{i-1}^e+\widecheck{b}^e\|^4+3\expec\|\bsc_{i}^e\|^4+\notag\\
&\qquad8\left(\expec\|(\cJ^e_{\epsilon}-\bcD_{22,i-1})\bcwc^e_{i-1}-\bcD_{21,i-1}\bcwb_{i-1}^e+\widecheck{b}^e\|^2\right)\left(\expec\|\bsc_{i}^e\|^2\right)\notag\\
&=\expec\left\|t\frac{1}{t}\cJ^e_{\epsilon}\bcwc^e_{i-1}-(1-t)\frac{1}{1-t}(\bcD_{22,i-1}\bcwc^e_{i-1}+\bcD_{21,i-1}\bcwb_{i-1}^e-\widecheck{b}^e)\right\|^4+3\expec\|\bsc_{i}^e\|^4+\notag\\
&\qquad8\left(\expec\left\|t\frac{1}{t}\cJ^e_{\epsilon}\bcwc^e_{i-1}-(1-t)\frac{1}{1-t}(\bcD_{22,i-1}\bcwc^e_{i-1}+\bcD_{21,i-1}\bcwb_{i-1}^e-\widecheck{b}^e)\right\|^2\right)\left(\expec\|\bsc_{i}^e\|^2\right)\notag\\
&\leq\frac{1}{t^3}\|\cJ^e_{\epsilon}\|^4\expec\|\bcwc^e_{i-1}\|^4+\frac{1}{(1-t)^3}\expec\|\bcD_{22,i-1}\bcwc^e_{i-1}+\bcD_{21,i-1}\bcwb_{i-1}^e-\widecheck{b}^e\|^4+3\expec\|\bsc_{i}^e\|^4+\notag\\
&\qquad8\left(\expec\|\bsc_{i}^e\|^2\right)\left(\frac{1}{t}\|\cJ^e_{\epsilon}\|^2\expec\|\bcwc^e_{i-1}\|^2+\frac{1}{1-t}\expec\|\bcD_{22,i-1}\bcwc^e_{i-1}+\bcD_{21,i-1}\bcwb_{i-1}^e-\widecheck{b}^e\|^2\right)\notag\\
&\leq\frac{1}{t^3}\|\cJ^e_{\epsilon}\|^4\expec\|\bcwc^e_{i-1}\|^4+\frac{27}{(1-t)^3}(\expec\|\bcD_{22,i-1}\|^4\|\bcwc^e_{i-1}\|^4+\expec\|\bcD_{21,i-1}\|^4\|\bcwb_{i-1}^e\|^4+\|\widecheck{b}^e\|^4)+3\expec\|\bsc_{i}^e\|^4+\notag\\
&\qquad8\left(\expec\|\bsc_{i}^e\|^2\right)\left(\frac{1}{t}\|\cJ^e_{\epsilon}\|^2\expec\|\bcwc^e_{i-1}\|^2+\frac{3}{1-t}(\expec\|\bcD_{22,i-1}\|^2\|\bcwc^e_{i-1}\|^2+\expec\|\bcD_{21,i-1}\|^2\|\bcwb_{i-1}^e\|^2+\|\widecheck{b}^e\|^2)\right)\notag\\
&{\leq}\frac{(\rho(\cJ_{\epsilon})+\epsilon)^4}{t^3}\expec\|\bcwc^e_{i-1}\|^4+\frac{27\mu^4\sigma^4_{22}}{(1-t)^3}\expec\|\bcwc^e_{i-1}\|^4+\frac{27\mu^4\sigma^4_{21}}{(1-t)^3}\expec\|\bcwb_{i-1}^e\|^4+\frac{27}{(1-t)^3}\|\widecheck{b}^e\|^4+3\expec\|\bsc_{i}^e\|^4+\notag\\
&\qquad8\left(\expec\|\bsc_{i}^e\|^2\right)\left(\frac{(\rho(\cJ_{\epsilon})+\epsilon)^2}{t}\expec\|\bcwc^e_{i-1}\|^2+\frac{3\mu^2\sigma^2_{22}}{1-t}\expec\|\bcwc^e_{i-1}\|^2+\frac{3\mu^2\sigma^2_{21}}{1-t}\expec\|\bcwb_{i-1}^e\|^2+\frac{3}{1-t}\|\widecheck{b}^e\|^2\right)
\end{align}
for any arbitrary positive number $t\in(0,1)$. In the last inequality we used relation (122) in~\cite{nassif2019adaptation}. Selecting $t=\rho(\cJ_{\epsilon})+\epsilon<1$, we arrive at:
\begin{align}
\label{eq: fourth order bcwbr}
\expec\|\bcwc^e_{i}\|^4&\leq(\rho(\cJ_{\epsilon})+\epsilon)\expec\|\bcwc^e_{i-1}\|^4+\frac{27\mu^4\sigma^4_{22}}{(1-t)^3}\expec\|\bcwc^e_{i-1}\|^4+\frac{27\mu^4\sigma^4_{21}}{(1-t)^3}\expec\|\bcwb_{i-1}^e\|^4+\frac{27}{(1-t)^3}\|\widecheck{b}^e\|^4+3\expec\|\bsc_{i}^e\|^4+\notag\\
&\qquad8\left(\expec\|\bsc_{i}^e\|^2\right)\left((\rho(\cJ_{\epsilon})+\epsilon)\expec\|\bcwc^e_{i-1}\|^2+\frac{3\mu^2\sigma^2_{22}}{1-t}\expec\|\bcwc^e_{i-1}\|^2+\frac{3\mu^2\sigma^2_{21}}{1-t}\expec\|\bcwb_{i-1}^e\|^2+\frac{3}{1-t}\|\widecheck{b}^e\|^2\right)
\end{align}
where $1-t=1-\rho(\cJ_{\epsilon})-\epsilon$. 

In order to bound the fourth-order noise terms $\expec\|\bsb_{i}^e\|^4$ and $\expec\|\bsc_{i}^e\|^4$ appearing in~\eqref{eq: fourth order bcwbc} and~\eqref{eq: fourth order bcwbr}, we first note from eq.~(58) in~\cite{nassif2019adaptation} that:
\begin{align}
\expec\|\bsb_{i}^e\|^4+\expec\|\bsc_{i}\|^4&\leq\expec(\|\bsb_{i}^e\|^2+\|\bsc_{i}\|^2)^2{=}\expec\|\mu(\cV_{\epsilon}^e)^{-1}\cA^e\bs_i^e\|^4\leq\mu^4v_1^4\expec\|\bs_i^e\|^4.\label{eq: fourth order intermediate 3}
\end{align}
Now, applying Jensen's inequality to the convex function $f(x)=x^2$, we can write:
\begin{equation}
\expec\|\bs_i^e\|^4=\expec(\|\bs_i^e\|^2)^2=4\expec\left(\sum_{k=1}^N\|\bs_{k,i}\|^2\right)^2=4\expec\left(\sum_{k=1}^N\frac{1}{N}N\|\bs_{k,i}\|^2\right)^2\leq 4N\expec\left(\sum_{k=1}^N\|\bs_{k,i}\|^4\right)=4N\sum_{k=1}^N\expec\|\bs_{k,i}\|^4,\label{eq: fourth order intermediate 2}
\end{equation} 
in terms of the individual gradient noise processes, $\expec\|\bs_{k,i}\|^4$. For each term $\bs_{k,i}$, we have from~\eqref{eq: gradient noise fourth moment condition} and from the Jensen's inequality applied to the convex norm $\|x\|^4$:
\begin{align}
\expec\|\bs_{k,i}(\bw_{k,i-1})\|^4&\leq(\beta_{4,k}/h)^4\expec\|\bw_{k,i-1}\|^4+\sigma^4_{s4,k}\notag\\
&=(\beta_{4,k}/h)^4\expec\|\bw_{k,i-1}-w^o_k+w^o_k\|^4+\sigma^4_{s4,k}\notag\\
&\leq8(\beta_{4,k}/h)^4\expec\|\bwt_{k,i-1}\|^4+8(\beta_{4,k}/h)^4\|w^o_k\|^4+\sigma^4_{s4,k}\notag\\
&\leq\bar\beta_{4,k}^4\expec\|\bwt_{k,i-1}\|^4+\bar{\sigma}^4_{s4,k}
\end{align}
where $\bar\beta_{4,k}^4=8(\beta_{4,k}/h)^4$ and $\bar{\sigma}^4_{s4,k}=8(\beta_{4,k}/h)^4\|w^o_k\|^4+\sigma^4_{s4,k}$. Using the relations,
\begin{equation}
\sum_{k=1}^N\|\bwt_{k,i-1}\|^4\leq(\|\bwt_{1,i-1}\|^2+\|\bwt_{2,i-1}\|^2+\ldots+\|\bwt_{N,i-1}\|^2)^2=\|\bcwt_{i-1}\|^4=\left(\frac{1}{2}\|\bcwt_{i-1}^e\|^2\right)^2=\frac{1}{4}\|\bcwt_{i-1}^e\|^4,\label{eq: fourth order intermediate}
\end{equation} 
\begin{equation}
\|(\cV_\epsilon^e)^{-1}\bcwt^e_{i-1}\|^4=(\|\bcwb_{i}^e\|^2+\|\bcwc^e_{i}\|^2)^2\leq2\|\bcwb_{i}^e\|^4+2\|\bcwc^e_{i}\|^4,\label{eq: fourth order intermediate 1}
\end{equation} 
the term $\expec\|\bs_i\|^4$ in~\eqref{eq: fourth order intermediate 2} can be bounded as follows:
\begin{align}
\expec\|\bs_i^e\|^4&\leq 4N\sum_{k=1}^N\bar\beta_{4,k}^4\expec\|\bwt_{k,i-1}\|^4+4N\sum_{k=1}^N\bar{\sigma}^4_{s4,k}\notag\\
&\leq 4 \beta^4_{4,\max}\sum_{k=1}^N\expec\|\bwt_{k,i-1}\|^4+\sigma^4_{s4}\notag\\
&\hspace{-1mm}\overset{\eqref{eq: fourth order intermediate}}{\leq}\beta^4_{4,\max}\expec\|\cV^e_{\epsilon}(\cV^{e}_{\epsilon})^{-1}\bcwt^e_{i-1}\|^4+\sigma^4_{s4}\notag\\
&\leq\beta^4_{4,\max}\|\cV^e_{\epsilon}\|^4\expec\|(\cV^{e}_{\epsilon})^{-1}\bcwt^e_{i-1}\|^4+\sigma^4_{s4}\notag\\
&\hspace{-1mm}\overset{\eqref{eq: fourth order intermediate 1}}{\leq}2\beta^4_{4,\max}v_2^4[\expec\|\bcwb_{i-1}^e\|^4+\expec\|\bcwc^e_{i-1}\|^4]+\sigma^4_{s4}
\end{align}
where $\beta^4_{4,\max}\triangleq N\max_{1\leq k\leq N}\bar\beta_{4,k}^4$, and $\sigma^4_{s4}\triangleq 4N\sum_{k=1}^N\bar{\sigma}^4_{s4,k}$. Substituting into~\eqref{eq: fourth order intermediate 3}, we get:
\begin{equation}
\expec\|\bsb_{i}^e\|^4+\expec\|\bsc_{i}\|^4\leq2\mu^4\beta^4_{4,\max}v_1^4v_2^4[\expec\|\bcwb_{i-1}^e\|^4+\expec\|\bcwc^e_{i-1}\|^4]+\mu^4v_1^4\sigma^4_{s4}.\label{eq: fourth order intermediate 4}
\end{equation}
Returning to~\eqref{eq: fourth order bcwbc}, and using the bounds~(128) in~\cite{nassif2019adaptation} and~\eqref{eq: fourth order intermediate 4}, we  find that:
\begin{align}
\label{eq: fourth order bcwbc 1}
\expec\|\bcwb_{i}^e\|^4&\leq(1-\mu\sigma_{11})\expec\|\bcwb_{i-1}^e\|^4+\frac{\mu\sigma_{12}^4}{\sigma_{11}^3}\expec\|\bcwc^e_{i-1}\|^4+6\mu^4\beta^4_{4,\max}v_1^4v_2^4[\expec\|\bcwb_{i-1}^e\|^4+\expec\|\bcwc^e_{i-1}\|^4]+\notag\\
&\quad3\mu^4v_1^4\sigma^4_{s4}+8\mu^2 v_1^2\beta^2_{\max}v_2^2(1-\mu\sigma_{11})(\expec\|\bcwb_{i-1}^e\|^2)^2+8\mu^2 v_1^2\beta^2_{\max}v_2^2(1-\mu\sigma_{11})\expec\|\bcwc^e_{i-1}\|^2\expec\|\bcwb_{i-1}^e\|^2+\notag\\
&\quad8\mu^2 v_1^2\sigma^2_s(1-\mu\sigma_{11})\expec\|\bcwb_{i-1}^e\|^2+8\mu^3\frac{\sigma_{12}^2}{\sigma_{11}} v_1^2\beta^2_{\max}v_2^2\expec\|\bcwb_{i-1}^e\|^2\expec\|\bcwc^e_{i-1}\|^2+\notag\\
&\quad8\mu^3\frac{\sigma_{12}^2}{\sigma_{11}} v_1^2\beta^2_{\max}v_2^2(\expec\|\bcwc^e_{i-1}\|^2)^2+8\mu^3\frac{\sigma_{12}^2}{\sigma_{11}} v_1^2\sigma^2_s\expec\|\bcwc^e_{i-1}\|^2.
\end{align}
Using the properties that, for any two random variables $\ba$ and $\bc$, it holds that~\cite[pp.~528]{sayed2014adaptation}:
$$(\expec \ba)^2\leq\expec\ba^2,\qquad 2(\expec \ba^2)(\expec \bc^2)\leq\expec \ba^4+\expec \bc^4,$$
we can write:
\begin{align}
2(\expec\|\bcwb_{i-1}^e\|^2)(\expec\|\bcwc^e_{i-1}\|^2)&\leq\expec\|\bcwb_{i-1}^e\|^4+\expec\|\bcwc^e_{i-1}\|^4,\\
(\expec\|\bcwb_{i-1}^e\|^2)^2&\leq\expec\|\bcwb_{i-1}^e\|^4,\\
(\expec\|\bcwc^e_{i-1}\|^2)^2&\leq\expec\|\bcwc^e_{i-1}\|^4,
\end{align}
so that:
\begin{equation}
\expec\|\bcwb_{i}^e\|^4\leq a\expec\|\bcwb_{i-1}^e\|^4+b\expec\|\bcwc^e_{i-1}\|^4+a'\expec\|\bcwb_{i-1}^e\|^2+b'\expec\|\bcwc^e_{i-1}\|^2+e
\end{equation}
where:
\begin{align}
a&=1-\mu\sigma_{11}+O(\mu^2),\quad b=O(\mu),\quad a'=O(\mu^2),\quad b'=O(\mu^3),\quad e=O(\mu^4).
\end{align}
Returning to~\eqref{eq: fourth order bcwbr} and using similar arguments, we can verify that:
\begin{align}
\label{eq: fourth order bcwbr 1}
\expec\|\bcwc^e_{i}\|^4&\leq(\rho(\cJ_{\epsilon})+\epsilon)\expec\|\bcwc^e_{i-1}\|^4+\frac{27\mu^4\sigma^4_{22}}{(1-\rho(\cJ_{\epsilon})-\epsilon)^3}\expec\|\bcwc^e_{i-1}\|^4+\frac{27\mu^4\sigma^4_{21}}{(1-\rho(\cJ_{\epsilon})-\epsilon)^3}\expec\|\bcwb_{i-1}^e\|^4+\notag\\
&\qquad\frac{27}{(1-\rho(\cJ_{\epsilon})-\epsilon)^3}\|\widecheck{b}^e\|^4+6\mu^4\beta^4_{4,\max}v_1^4v_2^4[\expec\|\bcwb_{i-1}^e\|^4+\expec\|\bcwc^e_{i-1}\|^4]+3\mu^4v_1^4\sigma^4_{s4}+\notag\\
&\qquad4(\rho(\cJ_{\epsilon})+\epsilon)\mu^2 v_1^2\beta^2_{\max}v_2^2[\expec\|\bcwb_{i-1}^e\|^4+\expec\|\bcwc^e_{i-1}\|^4]+8(\rho(\cJ_{\epsilon})+\epsilon)\mu^2 v_1^2\beta^2_{\max}v_2^2\expec\|\bcwc^e_{i-1}\|^4+\notag\\
&\qquad8(\rho(\cJ_{\epsilon})+\epsilon)\mu^2 v_1^2\sigma^2_s\expec\|\bcwc^e_{i-1}\|^2+\frac{12\mu^4 v_1^2\beta^2_{\max}v_2^2\sigma^2_{22}}{1-\rho(\cJ_{\epsilon})-\epsilon}[\expec\|\bcwb_{i-1}^e\|^4+\expec\|\bcwc^e_{i-1}\|^4]+\notag\\
&\qquad\frac{24\mu^4 v_1^2\beta^2_{\max}v_2^2\sigma^2_{22}}{1-\rho(\cJ_{\epsilon})-\epsilon}\expec\|\bcwc^e_{i-1}\|^4+ \frac{24\mu^4\sigma^2_{22}v_1^2\sigma^2_s}{1-\rho(\cJ_{\epsilon})-\epsilon}\expec\|\bcwc^e_{i-1}\|^2+ \frac{24\mu^4\sigma^2_{21}v_1^2\beta^2_{\max}v_2^2}{1-\rho(\cJ_{\epsilon})-\epsilon}\expec\|\bcwb_{i-1}^e\|^4+\notag\\
&\qquad \frac{12\mu^4\sigma^2_{21}v_1^2\beta^2_{\max}v_2^2}{1-\rho(\cJ_{\epsilon})-\epsilon}[\expec\|\bcwb_{i-1}^e\|^4+\expec\|\bcwc^e_{i-1}\|^4]+\frac{24\mu^4\sigma^2_{21}v_1^2\sigma^2_s}{1-\rho(\cJ_{\epsilon})-\epsilon}\expec\|\bcwb_{i-1}^e\|^2+\notag\\
&\qquad\frac{24\mu^2 v_1^2\beta^2_{\max}v_2^2}{1-\rho(\cJ_{\epsilon})-\epsilon}\|\widecheck{b}^e\|^2[\expec\|\bcwb_{i-1}^e\|^2+\expec\|\bcwc^e_{i-1}\|^2]+\frac{24\mu^2 v_1^2\sigma^2_s}{1-\rho(\cJ_{\epsilon})-\epsilon}\|\widecheck{b}^e\|^2,
\end{align}
so that, 
\begin{equation}
\expec\|\bcwc^e_{i}\|^4\leq c\expec\|\bcwb_{i-1}^e\|^4+d\expec\|\bcwc^e_{i-1}\|^4+c'\expec\|\bcwb_{i-1}^e\|^2+d'\expec\|\bcwc^e_{i-1}\|^2+f,
\end{equation}
where the coefficients $\{c,d,c',d',f\}$ have the following form:
\begin{align}
c&=O(\mu^2),\quad d=\rho(\cJ_{\epsilon})+\epsilon+O(\mu^2),\quad c'=O(\mu^4),\quad d'=O(\mu^2),\quad f=O(\mu^4).
\end{align}
Therefore, we can write
\begin{equation}
\label{eq: fourth order moment recursion}
\left[\begin{array}{c}
\expec\|\bcwb_{i}^e\|^4\\
\expec\|\bcwc^e_{i}\|^4
\end{array}\right]\preceq
\underbrace{\left[\begin{array}{cc}
a&b\\
c&d
\end{array}\right]}_{\Gamma}
\left[\begin{array}{c}
\expec\|\bcwb_{i-1}^e\|^4\\
\expec\|\bcwc^e_{i-1}\|^4
\end{array}\right]+
\left[\begin{array}{cc}
a'&b'\\
c'&d'
\end{array}\right]
\left[\begin{array}{c}
\expec\|\bcwb_{i-1}^e\|^2\\
\expec\|\bcwc^e_{i-1}\|^2
\end{array}\right]+
\left[\begin{array}{c}
e\\
f
\end{array}\right]
\end{equation}
in terms of the $2\times 2$ coefficient matrix $\Gamma$ of the form~(132) in~\cite{nassif2019adaptation}
which is stable matrix for sufficiently small $\mu$ and $\epsilon$. Moreover, using relation~(135) in~\cite{nassif2019adaptation}, we have:
\begin{equation}
\limsup_{i\rightarrow\infty}\left[\begin{array}{cc}
a'&b'\\
c'&d'
\end{array}\right]\left[\begin{array}{c}
\expec\|\bcwb_{i-1}^e\|^2\\
\expec\|\bcwc^e_{i-1}\|^2
\end{array}\right]=\left[\begin{array}{c}
O(\mu^3)\\
O(\mu^4)
\end{array}\right].
\end{equation}
In this case, we can iterate~\eqref{eq: fourth order moment recursion} and use relation~(134) in~\cite{nassif2019adaptation} to conclude that:
\begin{equation}
\limsup_{i\rightarrow\infty}\expec\|\bcwb_{i}^e\|^4=O(\mu^2), \qquad\limsup_{i\rightarrow\infty}\expec\|\bcwc^e_{i}\|^4=O(\mu^4),
\end{equation}
and, therefore,
\begin{align}
\limsup_{i\rightarrow\infty}\expec\|\bcwt_{i}^e\|^4&=\limsup_{i\rightarrow\infty}\expec\left\|\cV_{\epsilon}\left[\begin{array}{c}\bcwb_{i}^e\\\bcwc^e_{i}\end{array}\right]\right\|^4\notag\\
&\leq v_2^4\limsup_{i\rightarrow\infty}\expec(\|\bcwb_{i}^e\|^2+\|\bcwc^e_{i}\|^2)^2\notag\\
&\leq \limsup_{i\rightarrow\infty}2v_2^4(\expec\|\bcwb_{i}^e\|^4+\expec\|\bcwc^e_{i}\|^4)= O(\mu^2).
\end{align}
\section{Stability of the coefficient matrix $\cB$}
\label{sec: Stability of of the coefficient matrix B}
In this Appendix, we show that, under the same settings of Theorem~\ref{theo: performance error}, the constant matrix $\cB$ defined by~\eqref{eq: constant B} is stable for sufficiently small step-sizes. To establish this, we use similar argument as in~\cite{zhao2015asynchronous,sayed2014adaptation}. We first note that the matrix $\cB$ in~\eqref{eq: constant B} is similar to the matrix $\cBb$ in~\eqref{eq: cBb definition}, and therefore has the same eigenvalues as the block matrix $\cBb$ written as:
\begin{equation}
\cB\sim\left[\begin{array}{cc}
I_{hP}-\cD_{11}&-\cD_{12}\\
-\cD_{21}&\cJ_{\epsilon}^e-\cD_{22}
\end{array}\right],
\end{equation}
where the blocks entries $\{\cD_{mn}\}$ are given by~\eqref{eq: definition cD 11}--\eqref{eq: definition cD 22}. In a manner similar to the arguments used in the proof of Theorem~1 in~\cite{nassif2019adaptation}, we can verify that:
\begin{align}
\cD_{11}&=O(\mu),\qquad\cD_{12}=O(\mu),\label{eq: order of D11}\\
\cD_{21}&=O(\mu),\qquad\cD_{22}=O(\mu),\label{eq: order of D22}\\
\rho(I_{hP}-\cD_{11})&=1-\sigma_{11}\mu=1-O(\mu),\label{eq: spectral radius of I-D11}
\end{align}
where $\sigma_{11}$ is a positive scalar independent of $\mu$. Thus, we obtain:
\begin{equation}
\cB\sim\left[\begin{array}{cc}
I_{hP}-O(\mu)&O(\mu)\\
O(\mu)&\cJ_{\epsilon}^e+O(\mu)
\end{array}\right].
\end{equation}
Now recall that the matrix $\cJ_{\epsilon}^e$ defined in Table~\ref{table: variables table} is $h(M-P)\times h(M-P)$ and has a Jordan structure. We consider here the complex data case since the real data case can be easily deduced from the complex case by removing the block $(\cJ_{\epsilon}^*)^\top$. It can be expressed in the following upper-triangular form:
\begin{equation}
\cJ_{\epsilon}^e=\left[\begin{array}{ccc|ccc}
\lambda_{a,2}&&\cK&&\\
&\ddots&&&\\
&&\lambda_{a,L}&&\\\hline
&&&\lambda_{a,2}^*&&\cK\\
&&&&\ddots&\\
&&&&&\lambda_{a,L}^*
\end{array}
\right]
\end{equation}
with scalars $\{\lambda_{a,\ell},\lambda_{a,\ell}^*\}$ on the diagonal, all of which have norms strictly less than one, and where the entries of the strictly upper-triangular matrix $\cK$ are either $\epsilon$ or zero. It follows that:
\begin{equation}
\cJ_{\epsilon}^e+O(\mu)=\left[\begin{array}{ccc|ccc}
\lambda_{a,2}+O(\mu)&&\cK+O(\mu)&&\\
&\ddots&&&O(\mu)\\
O(\mu)&&\lambda_{a,L}+O(\mu)&&\\\hline
&&&\lambda_{a,2}^*+O(\mu)&&\cK+O(\mu)\\
&O(\mu)&&&\ddots&\\
&&&O(\mu)&&\lambda_{a,L}^*+O(\mu)
\end{array}
\right]
\end{equation}

We introduce the eigen-decomposition of the Hermitian positive-definite matrix $\cD_{11}$ and denote it by~\cite{zhao2015asynchronous,sayed2014adaptation}:
\begin{equation}
\cD_{11}\triangleq U_d\Lambda_d U^*_d 
\end{equation}
where $U_d$ is unitary and $\Lambda_d$ has positive diagonal entries $\{\lambda_k\}$; the matrices $U_d$ and $\Lambda_d$ are $hP\times hP$. Using $U_d$, we further introduce the following block-diagonal similarity transformation:
\begin{equation}
\cT\triangleq\diag\{\mu^{P/M}U_d,\mu^{(hP+1)/hM},\ldots,\mu^{(hM-1)/hM},\mu\}.
\end{equation} 
We now use~\eqref{eq: cBb definition} to get:
\begin{equation}
\cT^{-1}\cBb\cT=\left[
\begin{array}{c|cccc}
B&&O(\mu^{(hM+1)/hM})&&\\
\hline
&\lambda_{a,2}+O(\mu)&&{O(\mu^{1/hM})}\\
O(\mu^{P/M})&&\ddots&\\
&O(\mu^{(hM-1)/hM})&&\lambda^*_{a,L}+O(\mu)
\end{array}
\right]
\end{equation}
where we introduced the $hP\times hP$ diagonal matrix:
\begin{equation}
B\triangleq I_{hP}-\Lambda_d.
\end{equation}
It follows that all off-diagonal entries of the above transformed matrix are at most $O(\mu^{1/hM})$. Although the factor $\mu^{1/hM}$ decays slower than $\mu$, it nevertheless becomes small for sufficiently small $\mu$. Calling upon the Gershgorin's theorem\footnote{Consider an $N\times N$ matrix $A$ with scalar entries $\{a_{k\ell}\}$. With each diagonal entry $a_{kk}$ we associate a disc in the complex plane centered at $a_{kk}$ and with $r_k=\sum_{\ell=1,\ell\neq k}^{N}|a_{k\ell}|$. That is, $r_k$ is equal to the sum of the magnitudes of the non-diagonal entries on the same row as $a_{kk}$. We denote the disc by $D_k$; it consists of all points that satisfy $D_k=\{z\in\mathbb{C} \text{ such that }|z-a_{kk}|\leq r_k\}$. Gershgorin's theorem states that the spectrum of $A$ (i.e., the set of all its eigenvalues, denoted by $\lambda(A)$) is contained in the union of all $N$ Gershgorin discs $$\lambda(A)\subset\cup_{k=1}^ND_k.$$ A stronger statement of the Gershgorin theorem covers the situation in which some of the Gershgorin discs happen to be disjoint. Specifically, if the union of the $L$ discs is disjoint from the union of the remaining $N-L$ discs, then the theorem further asserts that $L$ eigenvalues of $A$ will lie in the first union of $L$ discs and the remaining $N-L$ eigenvalues of $A$ will lie in the second union of $N-L$ discs.}, we conclude that the eigenvalues of $\cB$ are either located in the Gershgorin circles that are centered at the eigenvalues of $B$ with radii $O(\mu^{(hM+1)/hM})$ or in the Gershgorin circles that are centered at the $\{\lambda_{a,\ell},\lambda_{a,\ell}^*\}$ with radii $O(\mu^{1/M})$, namely,
\begin{equation}
|\lambda(\cB)-\lambda(B)|\leq O(\mu^{(hM+1)/hM})\quad\text{or}\quad|\lambda(\cB)-\lambda_{a,\ell}+O(\mu)|\leq O(\mu^{1/hM})\quad\text{or}\quad|\lambda(\cB)-\lambda^*_{a,\ell}+O(\mu)|\leq O(\mu^{1/hM})
\end{equation}
where $\lambda(\cB)$ and $\lambda(B)$ denote any of the eigenvalues of $\cB$ and $B$, and $\ell=1,\ldots,L$. It follows that:
\begin{equation}
\label{eq: bound on rho(cB)}
\rho(\cB)\leq\rho(B)+O(\mu^{(hM+1)/hM})\quad\text{or}\quad\rho(\cB)\leq\rho(\cJ_{\epsilon})+O(\mu)+O(\mu^{1/hM}).
\end{equation}
Now since $\cJ_{\epsilon}$ is a stable matrix, we know that $\rho(\cJ_{\epsilon})<1$. We express this spectral radius as:
\begin{equation}
\rho(\cJ_{\epsilon})=1-\delta_J
\end{equation}
where $\delta_J$ is positive and independent of $\mu$. We also know from~\eqref{eq: spectral radius of I-D11} that:
\begin{equation}
\rho(B)=1-\sigma_{11}\mu<1,
\end{equation}
since $B=U^*_d(I_{hP}-\cD_{11})U_d$. We conclude from~\eqref{eq: bound on rho(cB)} that:
\begin{equation}
\rho(\cB)\leq1-\sigma_{11}\mu+O(\mu^{(hM+1)/hM})\quad\text{or }\rho(\cB)\leq1-\delta_J+O(\mu)+O(\mu^{1/hM}).
\end{equation}
If we now select $\mu\ll 1$ small enough such that:
\begin{equation}
\label{eq: conditions for stability of cB}
O(\mu^{(hM+1)/hM})<\sigma_{11}\mu,\quad\text{and}\quad O(\mu^{1/hM})+O(\mu)<\delta_J
\end{equation}
then we would be able to conclude that $\rho(\cB)<1$ so that $\cB$ is stable for sufficiently small step-sizes, as claimed. 

If we exploit the structure of $\cBb$ in~\eqref{eq: cBb definition} we can further show, for sufficiently small step-sizes, that:
\begin{align}
(I-\cB)^{-1}&=O(1/\mu)\\
(I-\cBb)^{-1}&=\left[\begin{array}{c|c}
O(1/\mu)&O(1)\\
\hline
O(1)&O(1)
\end{array}
\right]\label{eq: inverse of I-cBb}
\end{align}
where the leading $(1,1)$ block in $(I-\cBb)^{-1}$ has dimensions $hP\times hP$.

To establish this we first note that, by similarity, the matrix $\cBb$ is stable. Let 
\begin{align}
\cX=I-\cBb&=\left[\begin{array}{cc}
\cD_{11}&\cD_{12}\\
\cD_{21}&I-\cJ^e_{\epsilon}+\cD_{22}
\end{array}\right]\triangleq\left[\begin{array}{cc}
\cX_{11}&\cX_{12}\\
\cX_{21}&\cX_{22}
\end{array}\right],
\end{align}
where from~\eqref{eq: order of D11}--\eqref{eq: order of D22}, we have:
\begin{align}
\cX_{11}&=O(\mu),\qquad\cX_{12}=O(\mu),\label{eq: order of cX11}\\
\cX_{21}&=O(\mu),\qquad\cX_{22}=O(1).\label{eq: order of cX22}
\end{align}
The matrix $\cX$ is invertible since $I-\cBb$ is invertible. Moreover, $\cX_{11}$ is invertible since $\cD_{11}$ is Hermitian positive definite. Using the block matrix inversion formula
, we can write:
\begin{equation}
\cX^{-1}=\left[\begin{array}{cc}
\cX_{11}^{-1}+\cX_{11}^{-1}\cX_{12}\Delta^{-1}\cX_{21}\cX_{11}^{-1}&-\cX_{11}^{-1}\cX_{12}\Delta^{-1}\\
-\Delta^{-1}\cX_{21}\cX_{11}^{-1}&\Delta^{-1}
\end{array}\right]
\end{equation}
where $\Delta$ denotes the Schur complement of $\cX$ relative to $\cX_{11}$:
\begin{equation}
\Delta=\cX_{22}-\cX_{21}\cX_{11}^{-1}\cX_{12}=O(1).
\end{equation}
We then use~\eqref{eq: order of cX11}--\eqref{eq: order of cX22} to conclude~\eqref{eq: inverse of I-cBb}.

\section{Stability of first-order error moment of~\eqref{eq: long term error model 2}}
\label{sec: Stability of first-order error moment}
In this Appendix, we show that, under the same settings of Theorem~\ref{theo: performance error}, the first-order moment of the long-term model~\eqref{eq: long term error model 2} is stable for sufficiently small step-sizes, namely, it holds that:
\begin{equation}
\label{eq: order of the mean long-term model}
\limsup_{i\rightarrow\infty}\|\expec\bcwt^{e'}_i\|=O(\mu).
\end{equation} 

We first multiply both sides of recursion~\eqref{eq: mean error recursion of the long-term error model} from the left by $(\cV_{\epsilon}^e)^{-1}$ and use relation~(59) in~\cite{nassif2019adaptation} to get:
\begin{equation}
\label{eq: mean long term error model 2}
\underbrace{\left[\begin{array}{c}
\expec\bcwb^{e'}_{i}\\
\expec\bcwc^{e'}_{i}
\end{array}\right]}_{\triangleq y_i}=\underbrace{\left[\begin{array}{cc}
I_{hP}-\cD_{11}&-\cD_{12}\\
-\cD_{21}&\cJ^e_{\epsilon}-\cD_{22}
\end{array}\right]}_{\triangleq\cBb}\underbrace{\left[\begin{array}{c}
\expec\bcwb^{e'}_{i-1}\\
\expec\bcwc^{e'}_{i-1}
\end{array}\right]}_{\triangleq y_{i-1}}+\left[\begin{array}{c}
0\\
\widecheck{b}^e
\end{array}\right]
\end{equation}
where the matrix $\cBb$ in~\eqref{eq: cBb definition} is stable as shown in Appendix~\ref{sec: Stability of of the coefficient matrix B}. Recursion~\eqref{eq: mean long term error model 2} can be written more compactly as:
\begin{equation}
\label{eq: mean long term error model 3}
y_i=\cBb y_{i-1}+\left[\begin{array}{c}
0\\
\widecheck{b}^e
\end{array}\right].
\end{equation}
Since $\cBb$ is stable and $\widecheck{b}^e=O(\mu)$, we conclude from~\eqref{eq: mean long term error model 3} and~\eqref{eq: inverse of I-cBb} that:
\begin{align}
\label{eq: mean long term error model 4}
\lim_{i\rightarrow\infty}y_i&=(I-\cBb)^{-1}\left[\begin{array}{c}
0\\
\widecheck{b}^e
\end{array}\right]\overset{\eqref{eq: inverse of I-cBb}}=\left[\begin{array}{cc}
O(1/\mu)&O(1)\\
O(1)&O(1)
\end{array}\right]\left[\begin{array}{c}
0\\
O(\mu)
\end{array}\right]=O(\mu).
\end{align}
It follows that
\begin{equation}
\limsup_{i\rightarrow\infty}\left\|\left[\begin{array}{c}
\expec\bcwb^{e'}_{i}\\
\expec\bcwc^{e'}_{i}
\end{array}\right]\right\|=O(\mu),
\end{equation}
and, hence,
\begin{align}
\limsup_{i\rightarrow\infty}\|\expec\bcwt^{e'}_i\|&=\limsup_{i\rightarrow\infty}\left\|\cV_{\epsilon}^e\left[\begin{array}{c}
\expec\bcwb^{e'}_{i}\\
\expec\bcwc^{e'}_{i}
\end{array}\right]\right\|\leq\|\cV^e_\epsilon\|\left(\limsup_{i\rightarrow\infty}\left\|\left[\begin{array}{c}
\expec\bcwb^{e'}_{i}\\
\expec\bcwc^{e'}_{i}
\end{array}\right]\right\|\right)=O(\mu).
\end{align}

\end{appendices}

\bibliographystyle{IEEEbib}
{\balance{
\bibliography{reference}}}

\end{document}